\documentclass{article} 

\usepackage[dvipsnames]{xcolor}
\usepackage[margin=1in]{geometry}
\usepackage{amsmath,amsthm,amssymb} 
\usepackage{url}
\usepackage{forest}
\usepackage{tikz}
\usepackage{tikz-qtree}
\usepackage{wrapfig}
\usepackage{listings}
\usepackage{enumitem}
\usepackage{marginnote}
\usepackage{thm-restate}
\usepackage{mathtools}
\usepackage{microtype}
\usepackage[tableposition=above]{caption}
\usepackage{hyperref}

\usetikzlibrary{matrix,positioning, 
decorations.pathreplacing,
decorations.pathmorphing,
arrows,calc,shapes,fit,quotes}

\definecolor{bg}{gray}{1}
\definecolor{niceblue}{rgb}{.392,.584,.929}
\definecolor{dkgreen}{rgb}{0,0.6,0}
\definecolor{gray}{rgb}{0.5,0.5,0.5}
\definecolor{mauve}{rgb}{0.58,0,0.82}

\newtheorem{theorem}{Theorem}
\newtheorem{lemma}[theorem]{Lemma}

\newtheorem{observation}[theorem]{Observation}
\newtheorem{corollary}[theorem]{Corollary}

\theoremstyle{definition}
\newtheorem{definition}[theorem]{Definition}

\newcommand{\defn}{\emph}
\renewcommand{\epsilon}{\varepsilon}
\renewcommand{\hat}{\widehat}
\renewcommand{\tilde}{\widetilde}

\newcommand{\Smab}{\smash[b]{\binom{\log n}{k}}}

\newcommand{\problem}{Text Indexing with Mismatches}
\newcommand{\occ}{\# occ}

\newcommand{\poly}{\text{poly}}
\newcommand{\polylog}{\text{polylog}}
\newcommand{\calT}{\mathcal{T}}

\newcommand{\frakq}{\mathfrak{q}}
\newcommand{\predds}{$y$-fast trie}

\newcommand{\hi}[1]{\textcolor{red}{#1}}

\newif\ifsymbols
\symbolsfalse

\newif\iffull
\fulltrue

\title{Space-Efficient Text Indexing with Mismatches using Function Inversion}
\date{}
\author{
Jackson Bibbens\\
University of Massachusetts, Amherst, MA, USA\\
\texttt{jbibbens@umass.edu}
\and
Levi Borevitz\\
Northwestern University, Evanston, IL, USA\\
\texttt{leviborevitz2023@u.northwestern.edu}
\and
Samuel McCauley\\
Williams College, Williamstown, MA, USA\\
\texttt{sam@cs.williams.edu}
}

\begin{document}
\maketitle
 \begin{abstract}
	A classic data structure problem is to preprocess a string $T$ of length $n$ so that, given a query $q$, we can quickly find all substrings of $T$ with Hamming distance at most $k$ from the query string.  
    Variants of this problem have seen significant research both in theory and in practice.  
    For a wide parameter range, the best worst-case bounds are achieved by the ``CGL tree'' (Cole, Gottlieb, Lewenstein 2004), which achieves query time roughly $\tilde{O}(|q| + \log^k n + \occ)$, where $\occ$ is the size of the output, and space ${O}(n\log^k n)$.  The CGL Tree space was recently improved to $O(n \log^{k-1} n)$ (Kociumaka, Radoszewski 2026).

	A natural question that arises is whether a high space bound is necessary.  How efficient can we make queries when the data structure is constrained to $O(n)$ space?  
    While this question has seen extensive research, all known results have query time with unfavorable dependence on the alphabet size,
   $n$ and $k$. The state of the art query time from (Chan, Lam, Sung, Tam, Wong 2011) is roughly $\tilde{O}(|q| + |\Sigma|^k \log^{k^2 + k} n + \occ)$ for alphabet $\Sigma$.  

    We give an $O(n)$-space data structure with query time roughly $\tilde{O}(|q| + \log^{4k} n + \log^{2k} n \cdot \occ)$, with no dependence on the size of the alphabet.  Even for a constant-sized alphabet, this is the best known query time for linear space if $k\geq 3$ unless $\occ$ is large.  Our results give a smooth tradeoff between time and space.
     Interestingly, our results are the first to extend to the sublinear space regime: we give a succinct data structure using only $o(n)$ space in addition to the text itself, with only a modest increase in query time. 

	The main technical idea behind this result is to apply Fiat-Naor function inversion (Fiat, Naor 2000) to the CGL tree.  Combining these techniques is not immediate; in fact, we revisit the exposition of both the Fiat-Naor data structure and the CGL tree to obtain our bounds.  Along the way, we obtain improved performance for both data structures, which may be of independent interest.
\end{abstract}

\section{Introduction}
\label{sec:introduction}
In many string processing applications, it is crucial to be able to handle errors.
After all, real-world methodologies for obtaining and querying data are subject to mistakes and to noise.
With this motivation in mind, a classic string processing problem is finding substrings of a text that nearly match---but may not exactly match---a query string.

One common setting is building an indexing data structure.  We begin with a preprocessing step, in which we build an index on a text $T$.  After the preprocessing is complete, there is a sequence of queries $q$.  All characters in the text and any query are from an alphabet $\Sigma$.  The goal is to use the index to quickly find substrings of $T$ that are approximately equal to $q$.  
 In this work, 
 we use Hamming distance as our notion of approximate matching: we want to find all substrings of $T$ that differ from $q$ in at most $r$ locations for a specified parameter $r$.  
 Throughout this paper, we will assume that during preprocessing, we are given a bound $k$ on the maximum $r$ of any query.

This is a classic problem, first introduced by Minsky and Papert in 1969 as the ``best match'' or ``nearest neighbors'' problem~\cite{MinskyPapert69}.
In their original exposition, the problem is defined slightly differently: the input is a dictionary of strings rather than a single large text, and the goal is to find if any string in the dictionary is approximately equal to $q$; we discuss this further and explain why both versions are equivalent in Section~\ref{sec:dictionary_version}.
 We follow recent work (i.e.~\cite{ColeGoLe04,CohenAddadFeSt19}) in calling this the 
 \problem{} problem.   

Past work has found a number of efficient data structures for this problem.  A particularly performant data structure is the index of Cole, Gottlieb, and Lewenstein~\cite{ColeGoLe04}, which we refer to throughout this paper as the ``CGL tree.''  If $T$ has size $n$, and the maximum search radius is $k$, the CGL tree uses
 $O(n 3^k\binom{\log n + k}{k})$ space and can answer queries in\footnote{We use $\tilde{O}$ notation to suppress log factors: $\tilde{O}(f(n))$ is the same as $O(f(n)\polylog f(n))$.  Note that the log factors are in terms of the argument: for example, $\tilde{O}(\log^k n)$ suppresses $k$ and $\log\log n$ terms.}
$\tilde{O}(|q| + 6^k\binom{\log n + k}{k} + \occ)$ time, where $\occ$ is the number of substrings of $T$ that match $q$.  Perhaps most notably, these bounds exhibit optimal linear dependence on $|q|$ and $\occ$, with no dependence on the size of the alphabet. Furthermore, the search time is polylogarithmic for $k = O(1)$, since $\binom{\log n + k}{k} = O(\log^k n)$.

 The drawback of the Cole, Gottlieb, and Lewenstein result is the prohibitive space usage: 
 the CGL tree requires $\Theta(n3^k\binom{\log n+k}{k})$ words of space.
 There has been a long line of work on tradeoffs between space usage and query time~\cite{Cobbs95,HuynhHoLa06,LamSuWo08,ChanLaSu11}.  
 In particular, Chan et al.~\cite{ChanLaSu11} obtained an $O(n)$-sized index with query time $\tilde{O}(|q| + |\Sigma|^k (3\log n)^{k(k+1)} + \occ)$.

Recently, Kociumaka and Radoszewski~\cite{KociumakaRadoszewski26} were the first to improve the space of the CGL tree while maintaining essentially the same query time.  
 For large alphabets, their data structure requires $O(n \log^{k-1} n)$ space; roughly $\log n$ better than the CGL tree.  For small alphabets, they are able to save an additional log, achieving $O(n\log^{k-2 + \epsilon} n)$ space for any $\epsilon >0$ for sufficiently large $k$.  Their parameter setting is significantly different from ours: they improve the CGL tree space as much as possible while maintaining the same query time, whereas we increase the query time noticeably in exchange for drastic space improvements.  

The techniques used to obtain Kociumaka and Radoszewski's bounds in~\cite{KociumakaRadoszewski26} differ substantially from the techniques in this paper.  
Our strategy is to augment a truncated version of the CGL tree with a function-inversion data structure.  
Their approach leverages the underlying problem structure more directly.  
They first give a data structure whose query time depends on the length of a query; these bounds are useful if the length of a query is short.  For long queries, they carefully select ``anchors'' of the text to cover all matches; this improves space if the query is sufficiently long.  
\iffull
Using the better of these two data structures gives their final bounds.  
\fi
It is possible that their results could be augmented with a function-inversion data structure as described in this work, improving our tradeoffs by roughly a $\log n$ factor space.  We leave this to future work.

There has also been work on lower bounds for this problem.
Cohen-Addad, Feuilloley, and Starikovskaya~\cite{CohenAddadFeSt19} 
showed that in the pointer machine model, any data structure with $O(|q| + (\frac{\log n}{2k})^k + \occ)$ query time must use $\Omega(c^kn)$ space for some constant $c$.  
\iffull
Put in terms of previous results, they show that in the pointer machine model, it is not possible to simultaneously (1) substantially improve the query time of Cole, Gottlieb, and Lewenstein, and (2) obtain linear space for $k = \omega(1)$.
\fi

\begin{table*}[t]
    \centering
    \renewcommand{\arraystretch}{1.3}
    \captionsetup{position=below}
    \begin{tabular}{l l l}
        \hline
        \textbf{Space} & \textbf{Query Time} & \textbf{Reference} \\
        \hline
        \hline
        $O(n)$ words & $O(|q|^{k+1}|\Sigma|^k + \occ)$ & \cite{Cobbs95} \\
                     & $\tilde{O}(|q|^k|\Sigma|^k \log n + \occ)$ & \cite{HuynhHoLa06} \\
                     & $\tilde{O}(|q|^k|\Sigma|^k \log \log n + \occ)$ & \cite{LamSuWo08} \\
                     & $\tilde{O}(|q| +  |\Sigma|^k k^k 3^{k^2} \log^{k^2 + k} n + \occ)$ & \cite{ChanLaSu11} \\
                     & $\tilde{O}(|q| +  3^k \log^{4k} n + \log^{2k} n \cdot \occ)$ & Corollary~\ref{cor:result_linear_space} \\
        \hline
        $O(\frac{n}{\sigma} \log^k n)$ words; $\sigma \leq \log^k n$  & $\tilde{O}(|q| +  |\Sigma|^k k^k 3^{k^2} \sigma^k \log^k n + \occ)$ & \cite{ChanLaSu11} \\
        & $\tilde{O}(|q| + 3^k\sigma^3 \log^{k} n + \sigma^2 \cdot \occ)$ & Theorem~\ref{thm:result_large_space} \\
        \hline
        \hline
        $O(n \log|\Sigma|)$ bits  & $\tilde{O}(|q|^k|\Sigma|^k \log^2 n + \log n\cdot \occ )$ & \cite{HuynhHoLa06} \\
                     & $\tilde{O}(|q|^k|\Sigma|^k \log^{\epsilon} n + \log^\epsilon n\cdot  \occ )$ & \cite{LamSuWo08} \\
                     & $\tilde{O}(|q|\log^\epsilon n +  |\Sigma|^k k^k 3^{k^2} \log^{k^2 + k + \epsilon} n + \log^\epsilon n \cdot \occ)$ & \cite{ChanLaSu11} \\
        \hline
        $O(\frac{n}{\sigma} \log^k n)$ words; $\sigma \geq \log^{k} n$ & $\tilde{O}(|q| + 3^k \sigma^4 \log n + \sigma^3 \log^{1-k} n \cdot \occ)$ & Theorem~\ref{thm:result_small_space} \\
        
        \hline
    \end{tabular}
    \caption{Space-efficient data structures for \problem{}.  
    In this table, we simplify $\binom{x}{y} \approx x^y$;  this replacement does not affect the asymptotic bounds when $k=O(1)$.}
\label{table:related}
\end{table*}

\paragraph{Goal of This Work.}
 In the following discussion, we assume $k=O(1)$ to allow us to substitute $c^{\text{poly}(k)}\binom{\log n + k}{k} = \Theta(\log^k n)$.

A significant gap remains in the known time-space tradeoffs for this problem.
The best upper bound on the query time for a $\Theta(n)$ space data structure is $\tilde{O}(|\Sigma|^k\log^{k^2 + k} n + |q| + \occ)$, whereas the best lower bound is $\Omega(\log^k n + |q| + \occ)$.  Furthermore, known space-efficient upper bounds all have a $|\Sigma|^k$ term in the query time (see Table~\ref{table:related}),
whereas the query time of the CGL tree has no dependence on the alphabet size.

In this paper, we make significant progress towards closing this gap, obtaining a $\tilde{O}(\log^{4k} n + |q| + \log^{2k} n \cdot\occ)$ query time with $O(n)$ space. This is the first linear-space result with a query time that is a polynomial factor larger than the lower bound, and the first without a dependency on the size of the alphabet.  

Our time-space tradeoff also extends to a new parameter setting: a succinct data structure using \emph{sublinear} extra space.
Past work has used either $\Theta(n)$ words of space or $\Theta(n)$ bits of space~\cite{ChanLaSu11}.  We show that it is possible to achieve $o(n)$ bits of space in addition to the text $T$ while still retaining competitive query times.
This space savings is particularly well-motivated in the case that $T$ can be significantly compressed, leading to sublinear overall space usage.
At the extreme end of the time-space tradeoff, we maintain $o(n)$ query times while only storing $\tilde{O}(n^{3/4}\log^{3k} n)$ words of metadata in addition to the text $T$.  

\paragraph{Function Inversion for Data Structures.}

A classic problem in cryptography is function inversion: given a function $f: \{1,\ldots, N\}\rightarrow\{1,\ldots, N\}$, create a space-efficient data structure that allows us to quickly evaluate its inverse $f^{-1}$---that is to say, for some $j$, find an $i$ such that $f(i) = j$.  A space-inefficient solution is easy: we can store a lookup table for all possible queries in $\Theta(N)$ space; the goal is to obtain an $O(N/\sigma)$-space data structure for some $\sigma = \omega(1)$.  The best-known bounds for this classic problem are an $\tilde{O}(N/\sigma)$-space data structure that can evaluate $f^{-1}$ in $\tilde{O}(\sigma^3)$ time~\cite{FiatNaor00}. If $f$ is random, or if there are few pairs of elements $i,j$ with $f(i) = f(j)$, then the evaluation time can be improved to $\tilde{O}(\sigma^2)$ time~\cite{Hellman80, FiatNaor00}.

A recent line of work has shown that function inversion can help improve time-space tradeoffs for data structures.  The classic example is 
\iffull
3SUM indexing: suppose we want to preprocess an array $A$ so that given a query $q$, we can quickly find $i$ and $j$ such that $A[i] + A[j] = q$.  Thus, if we define a function $f(i,j) = A[i] + A[j]$, the query becomes equivalent to evaluating $f^{-1}(q)$.  With some extra care to ensure that the range and domain of $f$ are $\{1,\ldots, N\}$ for some integer $N$, it is possible to use this method to obtain a space-time tradeoff for 3SUM indexing; see~\cite{GolovnevGuHo20,KopelowitzPorat19,KirkpatrickKuMa26}.
\fi
 Further work has included function inversion results for other data structure problems including collinearity testing~\cite{AronovEzSh23}, string indexing~\cite{BilleGoLe24,CorriganGibbsKogan19}, and similarity search~\cite{McCauley24}.

The main idea behind this paper is to use function inversion to give a space-time tradeoff for \problem.  One interesting way this differs from past work is that we ``open the black box'' on both sides.  We must significantly alter both the CGL tree, and Fiat-Naor's function inversion result, in order to obtain our bounds.  Our exposition obtains improved bounds for each problem individually, and may be of independent interest.

\subsection{Results}
\label{sec:results}
 Our main result is the following time-space tradeoff.

\begin{theorem}
\label{thm:result_large_space}
    For any $\sigma$ satisfying $1 \leq \sigma \leq \Smab$, there exists a \problem{} data structure with $O(n\Smab/\sigma)$ space that can be constructed in $O(nk^2\Smab (\log n + k^2(\log\log n)^2))$ expected time and can answer queries in time 
    \[
    \tilde{O}\left(|q| + 3^k\sigma^3\Smab + \sigma^2 \cdot (\occ)\right).
    \]
\end{theorem}

This is a Las Vegas data structure: it is randomized, but the queries are always correct, and the space and query time bounds are worst-case. The randomization only affects preprocessing time.

We can immediately achieve linear space by setting $\sigma = \Smab$. 

\begin{corollary}
\label{cor:result_linear_space}
    There exists a \problem{} data structure with $O(n)$ space that can be constructed in expected time $O(nk^2\Smab(\log n + k^2(\log\log n)^2))$ 
    and can answer queries in time 
    \[    \tilde{O}\left(|q| + 3^k \binom{\log n}{k}^4 + \binom{\log n}{k}^2\cdot (\occ)\right).
    \]
\end{corollary}

Our data structure uses a new exposition of the CGL tree.
This exposition slightly improves the performance of the CGL tree: we improve the space from $O(3^kn \binom{\log n + k}{k})$ to $O(nk\Smab)$,  and the query time from $\tilde{O}(|q| + 6^k\smash[b]{\binom{\log n + k}{k}} + \occ)$ to $\tilde{O}(|q| + 3^k\Smab + \occ)$.   Specifically, we improve the base of the exponent (removing it entirely from the space term), which improves the asymptotics for $k = \omega(1)$, and we remove the additive $k$ term in the binomial, which improves the asymptotics for $k = \omega(\sqrt{\log n})$.
\begin{theorem}
\label{thm:result_improved_cgl}
    There exists a \problem{} data structure with $O(nk \Smab)$ space that can be constructed in expected time $O(nk^2\Smab\log n)$ and can answer queries in time 
    \[    \tilde{O}\left(|q| + 3^k \binom{\log n}{k} + \occ\right).
    \]
\end{theorem}

We also extend our results to give the first succinct data structure for \problem{}.   
This data structure has sublinear space in addition to the space required to store $T$.  This is particularly motivated in contexts when $T$ is compressed---our data structure can be built to match even a significantly sublinear-sized $T$.  Our succinct results are Monte Carlo, and queries are correct with high probability---that is to say, with probability $\geq 1 - 1/n^c$ for any constant $c$.

\begin{theorem}
\label{thm:result_small_space}
    For any $\sigma \geq \Smab$, 
    there exists a \problem{} data structure that uses $O(n\Smab/\sigma)$ space in addition to the text $T$, that can be constructed in 
    $O(nk^2\Smab + nk\sigma\log \sigma)$ 
    expected time and can answer queries correctly with high probability in time 
    \[
    \tilde{O}\left(|q| + 3^k\sigma^4 \log n + 
    \frac{\sigma^3}{\binom{\log n}{k}} \log n\cdot (\occ)\right).
    \]
\end{theorem}

Finally, we obtain slightly improved function inversion bounds.  For our results (and for many data structure applications), it is crucial that we obtain \emph{all} elements in the inverse $f^{-1}(j)$, whereas classic function inversion results~\cite{FiatNaor00,Hellman80} give a single element of $f^{-1}(j)$.  We show that if the elements we are inverting\footnote{The theorem as stated assumes one high-indegree element; to generalize, we can store all elements with indegree~$> \sigma$ in $O(n/\sigma)$ space and map them to $\bot$ manually.} have $|f^{-1}(j)|\leq \sigma$, a few minor changes to the original Fiat-Naor data structure are sufficient to find all elements in the inverse with probability $1$.  Furthermore, we can avoid all $\log n$ terms in the space and running time bound, using only $\log \sigma$ terms instead.

\begin{theorem}
\label{thm:function_inversion_all}
For any $f: [n]\rightarrow [n-1]\cup \bot$,  such that: 
\iffull
\begin{enumerate}[topsep=1pt, noitemsep]
\item $f(i)$ can be evaluated in $E(f)$ time for all $i$, and 
\item $|f^{-1}(j)| \leq \sigma$ for all $j\neq \bot$, 
\end{enumerate}
\fi
 there exists a data structure 
    using $O(n/\sigma)$ space
    such that for any $j\neq \bot$, we can find $f^{-1}(j)$ in 
    $O(\sigma^3 \cdot (E(f) + \log\sigma) \cdot \log^4 \sigma)$ time.  The data structure requires $O(n \cdot (E(f) + \log \sigma)\cdot\log \sigma)$ expected preprocessing time.
\end{theorem}

Past work on data structures for function inversion~\cite{McCauley24,AronovCaDa24} has found all elements in the preimage with sampling- or binary-search-based techniques on a black-box Fiat-Naor data structure, leading to (ignoring extra $\log n$ terms in~\cite{FiatNaor00}) a query time of $O(E(f) \cdot \sigma^4\log n)$.

\paragraph{A Note on Log Factors.}
We point out that the $\tilde{O}$ notation in our bounds does not suppress $\log n$ factors.
Specifically, in this paper we use $\tilde{O}(f(n))$ as a shorthand for $O(f(n)\cdot \polylog (f(n)))$.
Substituting the query time from Theorems~\ref{thm:result_large_space}, \ref{thm:result_improved_cgl}, or~\ref{thm:result_small_space} for $f(n)$, we are suppressing
terms polynomial in $k$, $\log |q|$, $\log \sigma$, or $\log\log n$.
These terms  are logarithmic compared to the query time,  
but $\log n$ terms are not---we need to track any $\log n$ terms that appear in the analysis, and design our data structures to avoid stray $\log n$ terms.  

\paragraph{Comparing Our Results to Previous Space-Efficient Work.}
We give a full comparison of previous space-efficient results in Table~\ref{table:related}.  
We give 
the first polynomial tradeoff between time and space: a space savings of $\sigma$ results in a $\sigma^3$ increase in query time, compared to the previous state of the art  $\sigma^k$. 
We are also the first space-efficient results without a $|\Sigma|^k$ term---in fact, our results have no dependence on $\Sigma$ so long as each character fits in a machine word.
We point out that $4k\leq k^2 + k$ for $k \geq 3$, so the logarithmic term in the query time alone is also improved for most values of $k$.

The one area where we do not improve on past work is that the final term in our query time is $\sigma^2 \cdot \occ$, whereas past work generally has an additive $\occ$ or $\occ \cdot \log n$ term.  This seems to be a downside of our approach: function inversion has, inherently, a high cost per element inverted (and therefore a high cost per element returned); see also~\cite{BilleGoLe24}.  That said, our gains elsewhere are significant enough that we still achieve improved performance unless the output size is large.  In particular, we give the best known query time for a linear-space data structure if
$\occ \leq |\Sigma|^k k^k 3^{k^2} \log^{k^2 - k} n$.  

We are the first to give succinct (i.e., $o(n)$-space) results, significantly generalizing past research attaining $O(n)$ bits of space.  To our knowledge, a succinct data structure was not known even for the $k=1$ case.

\subsection{Technical Overview}
\label{sec:technical_overview}
\paragraph{The CGL Tree and Function Inversion.}  

The main idea of our results is to use function inversion to reduce the space necessary to store the CGL tree.  The CGL tree, in short, is a combination of a trie and a binary search tree.  Each node in the CGL tree has a suffix of $T$ as a ``pivot''.  As in a binary search tree, the query algorithm guarantees that if a substring $s$ of $T$ has distance at most $k$ from $q$, then $s$ is a prefix of a pivot of a node traversed by $q$.

The CGL tree, which we refer to as $\calT_0$, has $O(n\Smab)$ leaves.  We truncate the CGL tree to obtain a subtree with $O(n\Smab/\sigma)$ leaves, each of which is the root of a subtree of size $O(\sigma)$ in the original tree.  We denote the truncated tree as $\calT$.  

Consider the following strategy that is space- and time-inefficient, but maintains correctness: for each leaf of the truncated tree $\calT$ we can keep a list of its $O(\sigma)$ descendants in $\calT_0$.  On a query, we begin by traversing $\calT$ exactly as we would have traversed $\calT_0$.  When we reach the leaf, we can compare the query to all $\sigma$ descendants.

To use function inversion, consider applying a \emph{label} to each leaf of the truncated tree $\calT$.  We define a set of functions $f_{1}, \ldots, f_{\binom{\log n}{k}}$ such that if the $i$th suffix of $T$ is a descendant of the leaf with label $\ell$, there exists a $\lambda$ such that $f_{\lambda}(i) = \ell$.

With this setup, we no longer need to store a list at each leaf of $\calT$.  Instead, we store a function inversion data structure on each $f_{\lambda}$.  When the query reaches a leaf $\ell$, rather than scanning through a list, we can calculate $f^{-1}_{\lambda}(\ell)$.  This gives exactly the set of strings that were stored in the list.

Performing a function inversion query requires $O(\sigma^3)$ time using the strategy of Fiat and Naor~\cite{FiatNaor00}.  This leads to query time roughly $O(\sigma^3\Smab)$.  The truncated tree $\calT$ and the function inversion data structure each require $O(n/\sigma)$ space.  

\paragraph{Challenges to This Strategy.}
The basic strategy above is to combine~\cite{ColeGoLe04} and~\cite{FiatNaor00}, but their results do not work together immediately.  We need to carefully address the following crucial details.
\begin{itemize}[topsep=1pt, noitemsep]
\item We must find a truncated tree $\calT$ with $O(|\calT_0|/\sigma)$ leaves, each with $O(\sigma)$ descendants.
\item We need to define the sequence of functions $f_1, \ldots, f_{\binom{\log n}{k}}$.
\item We need to ensure that we can use the function-inversion data structure to recover \emph{all} points stored in a given leaf, with probability $1$ (Fiat-Naor inversion only returns one).
\item We need a way to efficiently find the distance between the query $q$ and a substring of $T$.
\item We must avoid losing stray $\log n$ or $2^k$ terms while addressing these challenges.
\end{itemize}

To address these challenges, we ``open the black box'' on both sides.  We give a new exposition of the CGL tree, and an exposition of a simplified Fiat-Naor function inversion data structure.  

\paragraph{Changes to the CGL Tree.}

In the original exposition of~\cite{ColeGoLe04}, the starting point is a compressed trie on $T$, with two crucial changes.  First, a centroid path decomposition is calculated on the tree.  For each node, a new child is created, consisting of the union of all non-centroid path children of the node.  This union allows for more efficient searching when the query has errors at this point.  
Then, the children of each node, and the nodes along each path, are each merged into a BST-like data structure; the ``group tree.''   The group tree ensures that the height of the tree is at most $\log n$.  Stitching these data structures together gives their $O(c^k \smash[b]{\binom{\log n + k}{k}})$ query time.

Our exposition proceeds essentially in reverse.  
We begin with a BST-like structure: at each recursive node, we partition into 4 sets, each of at most half the size.  Then, to handle errors, rather than grouping trie nodes, we recursively alter both the query $q$ and the suffixes of $T$.  Each suffix of $T$ is altered exactly once.  We then store
8 children of each node:\footnote{In fact we will store 7 sets, as one of the altered sets will be unnecessary.} 
4 for the sets of altered suffixes, and 4 for the sets of unaltered suffixes.

The advantage of this change is that we have one recursive structure, rather than a combination of suffix tree nodes and two types of group tree nodes.  This gives us our query time improvements: in short, the original exposition incurred a constant-factor loss in performance at most $k$ times along each root-to-leaf path when the different structures were combined, leading to an extra $c^k$ term in time and space, and a $\smash[b]{\binom{\log n + k}{k}}$ term compared to our $\Smab$. 

The disadvantage of this change comes from difficulties in removing a final $\log n$ term.  Our goal is $\tilde{O}(\log^k n)$ performance: in particular, when the recursive search radius $r = 0$, we want performance $O(k\log\log n)$, so we do not have time to traverse a tree of height $\Omega(\log n)$.  In~\cite{ColeGoLe04} this is handled by combining a compressed trie on $T$ with an efficient predecessor data structure.  The challenge is that we have recursively changed the suffixes and the query---the suffix tree of $T$ might not reflect these recursive alterations.

Therefore, we must generalize the data structure in~\cite{ColeGoLe04} to handle alterations.  Our data structure has the same basic ideas as theirs, but our simplicity in the recursive exposition comes at the cost of significantly increased bookkeeping here.  
This generalization is the topic of Section~\ref{sec:improving_logs}.

\iffull
\paragraph{Black Boxes and Saving Logs.}
 As discussed above, in this paper we ``open the black box'': we give an exposition of the entire data structure to achieve the bounds of the CGL tree, as well as the data structure to achieve the Fiat-Naor function inversion bounds, rather than simply referencing their ideas.  A reader might reasonably wonder if reiterating these well-known results is necessary---or how using them as a black box would affect performance.   Here we discuss why these changes were in fact necessary to achieve the bounds in this paper, and give details of how a black box exposition would reduce performance.

First, we explain why it is necessary to change the CGL tree.  The main difference between our exposition and the original exposition in~\cite{ColeGoLe04} is that we have a single recursive process; their exposition alternates between ``group trees'' (a similar recursive process to ours), and trie-like nodes which may have up to $|\Sigma|$ children.  
We can therefore make two new guarantees: all nodes in our exposition have constant degree, and each child has at most half the descendants of its parent.  This has an immediate impact on combining with function inversion: we can truncate the tree to $O(n/\sigma)$ nodes which each have $O(\sigma)$ descendants in the worst case.   The upper bound on the number of descendants of a leaf is crucial to our analysis.

In addition to addressing the above correctness issues, we obtain somewhat improved bounds for the CGL tree, an improvement which is significantly magnified by the space-time tradeoff. Namely, the CGL tree originally had a $O(6^k\binom{\log n+k}{k})$ term in the query time, and $O(n3^k\binom{\log n+ k}{k})$ space.  This means that to obtain $O(n)$ space, we would need to set $\sigma \geq 3^k \tbinom{\log n + k}{k}$.  The query time with this value of $\sigma$ then has a term 
$6^k\cdot 3^{3k} \tbinom{\log n+k}{k}^4 = 162^k \tbinom{\log n + k}{k}^4$.  
This can be compared to $3^k\tbinom{\log n}{k}^4$ in Corollary~\ref{cor:result_linear_space}.
Saving the $c^k$ terms in the time and space---arguably minor terms in the original exposition---leads to a very noticeable difference in our final bounds, even for relatively small $k$.

Now, we explain why it is necessary to change Fiat-Naor's results, rather than using their results as a black box.
There are two challenges we need to overcome.
    First, Fiat-Naor does not try to save log factors.  Achieving our bounds means improving their results by a polylog factor in a number of places.  Most notably, their bounds lose a $\Theta(\log^2 n)$ factor in space, which would make the query time in Corollary~\ref{cor:result_linear_space} a $\Theta(\log^6 n)$ factor larger.
    Second, we need \emph{all} elements in the inverse of a function, and furthermore, we want to always find all of them with probability 1---in contrast, they find a single one with constant probability.

Surprisingly, simply omitting their method for handling high-indegree elements makes significant progress on addressing these problems.
 First, their $\Theta(\log^2 n)$ loss in space is due to sampling high-indegree elements of $f$.  Because we can prove that each leaf has $O(\sigma)$ descendants, we do not have high-indegree elements,\footnote{We will see that we in fact have one high-indegree element---but since we know it up front we can handle it explicitly.} so we can simply skip that part of their data structure.  
    This means both removing the table of sampled elements they store, and avoiding a subroutine where they repeatedly hash to avoid elements stored in the table.
    Secondly, we show that because there are no high-indegree elements, their result (essentially as-is) finds \emph{each} element in the inverse with constant probability.  
    
    Past work has used either sampling techniques~\cite{McCauley24} or a binary search-like strategy~\cite{AronovCaDa24} to find all elements.  
    Sampling and binary search both lead to an extra $\log n$ term---but on top of this, both techniques require (in short) roughly $O(\sigma^3)$ time \emph{per element} in the inverse, for $O(\sigma^4 \log n)$ time overall.  Our result shows that this bound can be improved in a special case: so long as all elements have indegree at most $\sigma$, a simplified version of the original function inversion algorithm achieves $O(\sigma^3)$ total time.
    
Finally, we want to ensure that we find {all} elements in the inverse with probability $1$ (rather than finding each with constant probability); this does require an addition to their method.
The usual strategy, see e.g.~\cite{FiatNaor00,McCauley24,KopelowitzPorat19,GolovnevGuHo20} is to repeat the Fiat-Naor query strategy $\log n$ times to boost the probability; however, we cannot afford that many repeats.  Instead, we repeat $\Theta(\log \sigma)$ times to guarantee that each element can be found correctly by the function inversion process with probability $1 - 1/\sigma$. Then, we store all elements not found by the function inversion process in a dictionary to guarantee correctness.

\fi

\subsection{Related Work.}

\paragraph{Known Lower Bounds.}
Cohen-Addad et al.~\cite{CohenAddadFeSt19} showed that there is a constant $c > 1$ such that any pointer-machine data structure for \problem{} with $O((\smash[b]{\frac{\log n}{2k})}^k + \occ)$ query time must have space $\Omega(c^k n)$---in particular, such a query time cannot be obtained for linear space if $k = \omega(1)$.  
For sufficiently small $\occ$, our results are the first to give $O(\log^{O(k)} n)$ query time with $O(n)$ space, within a polynomial factor of this lower bound.

However, these bounds are not entirely comparable to ours.  First, our data structure is not in the pointer machine model, as function inversion uses a lookup table.  Second, the lower bounds of~\cite{CohenAddadFeSt19} 
\iffull
(and in particular, the range reporting bound of Afshani~\cite{Afshani12} used as the basis of the bound)
\fi
heavily rely on the query time being linear in the output, e.g.\ they rely on query performance of the form $O(Q(n,k) + \occ)$.  The bounds in this paper, on the other hand, have a significant cost per string being output.  This downside seems intrinsic to the function inversion technique, at least using known techniques---space-efficient function inversion seems to require $\omega(1)$ time per element returned (see also~\cite{BilleGoLe24}).

\paragraph{Function Inversion.}
 Function inversion was first studied by Hellman~\cite{Hellman80}, who showed that a \emph{uniform random} function $f:[N]\rightarrow[N]$ can be inverted in $\tilde{O}(N/\sigma)$ space and $\sigma^2$ time.  Fiat-Naor generalized this result to general $f$, in exchange for an increase in cost: any $f:[N]\rightarrow[N]$ can be inverted in $\tilde{O}(N/\sigma)$ space and $\sigma^3$ time.  Recently Golovnev et al.\ improved this result for large $\sigma$ (along with simplifying the previous analysis),  inverting $f$ in $\tilde{O}(N/\sigma)$ space and $O(\min\{\sigma^3, \sigma N^{1/2}\})$ time~\cite{GolovnevGuHo20}.  We note that this final result uses either shared randomness or a non-uniform algorithm; therefore, we do not use their result here.

\paragraph{Using Function Inversion for Data Structures.}
Function inversion was originally applied to achieve time-space tradeoffs for an indexing variant of the classic 3SUM problem~\cite{GolovnevGuHo20,KopelowitzPorat19}.
More recently, these ideas have been extended to the idea of collinearity testing~\cite{AronovEzSh23}.  Aronov et al.~\cite{AronovCaDa24} give a toolbox for using function inversion for a variety of data structure primitives.  

The work of Bille et al.~\cite{BilleGoLe24} also applies function inversion to a string indexing problem.  Specifically, they solve the \emph{gapped string indexing problem}: given two strings $P_1$ and $P_2$, find instances of $P_1$ and $P_2$ in $T$ with distance (in this context: the difference in the index where each starts) within a given range $[\alpha, \beta]$.  In terms of techniques, their result differs significantly from ours: their main technical contribution is a data structure that allows a reduction to a sequence of 3SUM indexing queries, after which it is possible to use the results of~\cite{GolovnevGuHo20} as a black box to solve their problem.

Corrigan-Gibbs and Kogan~\cite{CorriganGibbsKogan19} apply function inversion to obtain new time-space tradeoffs for the ``systematic substring-search problem'', which is essentially a special case of \problem{} with $k=0$, $\Sigma = \{0,1\}$ and $|q| \leq \log n$.   Interestingly, they show that this problem is \emph{equivalent} to function inversion: any improvement on one implies improvement for the other. 

The work of McCauley~\cite{McCauley24} also applies function inversion to a nearest neighbor problem; specifically, high-dimensional approximate similarity search under the $\ell_2$ metric.  At a high level, the techniques of their work are similar: they trim a tree to use less space, and then use function inversion to recover the missing pieces.  
\iffull
However, the different setting leads to significantly different technical challenges---in particular, much of their result focuses on setting internal space-saving parameters, and then combining a tree-like approach with a (purely) hashing-based approach; neither of these is an option in our setting.
\fi

\subsection{Paper Outline}
In Section~\ref{sec:preliminaries} we give a detailed description of our model and notation.
We also give a data structure to help us compare suffixes of $T$ and $q$ more quickly, even suffixes that we have altered during recursive calls.
 In Section~\ref{sec:cgl} we give our exposition of the CGL tree.
Then, in Section~\ref{sec:space_efficient_cgl}, we describe how to achieve the space-efficient results given in Section~\ref{sec:results}, including an exposition of how we can adjust the Fiat-Naor function inversion strategy to obtain Theorem~\ref{thm:function_inversion_all}.
 Finally, in Section~\ref{sec:improving_logs}, we show how to traverse our data structure quickly, avoiding an extra $\log n$ term in Theorems~\ref{thm:result_large_space} and~\ref{thm:result_improved_cgl}.

\section{Preliminaries}
\label{sec:preliminaries}

\paragraph{String Notation and Definitions.}
The strings in this paper are assumed to consist of characters from a set $\Sigma$, each of which fits in $\Theta(1)$ machine words.
For a string $s$, and for any $i\in \{1, \ldots, |s|\}$, we use $s[i]$ to denote character $i$ of $s$; we $1$-index strings.

 We define the \defn{Hamming distance} of two strings $s_1$ and $s_2$ to be the number of indices $i\in \{1, \ldots, \min(|s_1|, |s_2|)\}$ such that $s_1[i]\neq s_2[i]$.  We denote this as $d_H(s_1, s_2)$.  
Note that this is the classic definition of Hamming distance if $|s_1| = |s_2|$; if the sizes are different, any further characters in the longer string are ignored.
This generalized notion of Hamming distance will make our exposition slightly easier.  For example, we can rephrase \problem{} as: find all suffixes of $T$ with distance at most $r$ from $q$.

We append to $T$ $2k+1$ copies of a special character $\$$ that does not appear in $q$ or $T$.  When we refer to ``suffixes'' of $T$, we only refer to the $n$ suffixes that do not start with $\$$.  With these special characters appended, no suffix of $T$ is a prefix of another, even if $k$ characters of each suffix are changed.  This has an added benefit that with our definition of distance, all suffixes of $T$ with length less than $|q|$ have distance greater than $k$ from $q$.

The \defn{longest common prefix} between strings $s_1$ and $s_2$, denoted $LCP(s_1,s_2)$,  is  the largest $i$ such that for all $j\leq i$, $s_1[j] = s_2[j]$. If $s_1$ is a prefix of $s_2$ or vice versa then $LCP(s_1, s_2) = \min\{|s_1|,|s_2|\}$.  

\paragraph{Altered Strings.}
Our data structure is recursive.  While recursing, we will make character substitutions on both the query string and suffixes of $T$.  We will keep track of the characters that have been changed using \defn{alterations}, which we define here.

An \defn{alteration} consists of an index and a character.
 A \defn{$k$-altered string} consists of a string $s$ and a sequence of \emph{at most} $k$ alterations $a$; we also call the string \defn{altered} without specifying $k$.  
 If $s$ is the query $q$, we refer to $(s,a)$ as an \defn{altered query}.
The resulting string
can be obtained by starting with $s$, and, for each alteration in $a$ consisting of an index $i$ and character $c$, replacing $s[i]$ with $c$.
We emphasize that $k$ is an upper bound on the number of alterations: much of our data structure, built with maximum search radius $k$, consists of $k$-altered suffixes of $T$. This means that each suffix has at most $k$ alterations, but could have fewer.
We denote an altered string using the notation $(s, a)$.

\paragraph{Other Notation and Definitions.}
For a function $f$, we use $f^{-1}(y)$ to denote the set of all $x$ such that $f(x) = y$.

 All logs in this paper are base $2$.
All bounds in this paper are in the word RAM model with words of size $\Omega(\log n)$, and all space bounds are in the number of required words.

We use $[x]$ to denote the set $\{1, \ldots, x\}$ for any $x$.

We use $n$ to denote the length of $T$, $k$ the maximum search radius of the data structure, and $\sigma$ the space-saving parameter: we want $O(n\Smab/\sigma)$ space.  For Theorem~\ref{thm:result_small_space}, our goal is $o(n)$ space, i.e.\ $\sigma > \Smab$.  In this case, it is often useful to use a variable $\tau = \lceil\sigma/\Smab\rceil$; that way, the space bound can be rewritten $O(n/\tau)$.

Our results assume that any character of $T$ can be accessed in $O(1)$ time.  If $T$ is stored in such a way that a longer time is required (e.g., if it is compressed) the bounds still hold with proportionally increased preprocessing and query time.

\iffull
For Theorems~\ref{thm:result_large_space} and~\ref{thm:result_small_space} the query time is $\Omega(n)$---and can therefore be accomplished with a linear scan---if $\sigma > n^{1/3}$ or $k > (\log n)/2$; with this in mind we assume that $\sigma < n^{1/3}$ and $k < (\log n)/2$ for the remainder of the paper.  Similarly, we assume $k \geq 1$ throughout, as otherwise the problem reduces to a dictionary lookup and can be solved with known techniques.
\fi

\subsection{First \texorpdfstring{$k+1$}{k+1} Mismatches Data Structure}
\label{sec:lcp}

Our results rely on a crucial subroutine, the \defn{first $k+1$ mismatches}: given two strings $\hat{s_1}$ and $\hat{s_2}$, where $\hat{s_2}$ is a substring of $T$ with at most $k+1$ alterations, and $\hat{s_1}$ is either a substring of $q$ with at most $k+1$ alterations or a substring of $T$ with at most $k+1$ alterations, find the first $k+1$ indices $i$ where $\hat{s_1}[i] \neq \hat{s_2}[i]$.    
We fix $\hat{s_1}$ and $\hat{s_2}$ for the remainder of this section.
 Let $\hat{s_1} = (s_1, a_1)$ and $\hat{s_2} = (s_2, a_2)$, where $s_1$ is a substring of $T$ or $q$, and $s_2$ is a substring of $T$.  

This subroutine generalizes two queries we need for our data structure: (1) finding the first place where $\hat{s_1}$ and $\hat{s_2}$ differ (i.e., computing $LCP(\hat{s_1}, \hat{s_2})$), and (2) checking if $d_H(\hat{s_1}, \hat{s_2})\leq k$.

Note that the strings are $(k+1)$-altered rather than $k$-altered; this is for Section~\ref{sec:improving_logs}. While Sections~\ref{sec:cgl} and~\ref{sec:space_efficient_cgl} only use $k$-altered strings, Section~\ref{sec:improving_logs} uses $(k+1)$-altered strings.

Our strategy for this data structure is to use, as a black box, a known data structure to find the longest common prefix between two \emph{unaltered} suffixes: either two suffixes of $T$, or a suffix of $T$ and a suffix of $q$.  We then recursively alter each string to compute the first $k+1$ mismatches between the altered strings.   This strategy is similar to the strategy referred to as ``kangaroo jumps'' in the literature; see~\cite{NicolaeRajasekaran17}. 
 For $\tau = 1$ a similar result (using a similar strategy) was given independently in~\cite[Lemma 3.3]{KociumakaRadoszewski26}.

\paragraph{Finding the Mismatches.}
We first find the LCP without taking alterations into account: that is to say, we find $LCP(s_1, s_2)$ (we will specify how to do this momentarily).
If the LCP is before the first alteration in $a_1$ or $a_2$, we have found a place where $\hat{s_1}$ and $\hat{s_2}$ differ, and we append it to the list of mismatches.  If it is at or after the first alteration to one of the strings, we can check if the strings differ at that location manually, appending to the list if so.  Then, we recurse on the suffix of each string after the first alteration.

Since each string has $\leq k+1$ alterations, and we only list out $\leq k+1$ mismatches, we recurse at most $3k+3 = O(k)$ times.

\paragraph{Longest Common Prefix Data Structures.}
We use one of the following two data structures to find $LCP(s_1, s_2)$.

First, consider the case where the total space used is $\Omega(n)$, as in Theorem~\ref{thm:result_large_space}. This means we can use $O(n)$ space for the longest common prefix data structure.  In this case, Cole, Gottlieb, and Lewenstein~\cite{ColeGoLe04} show that it is possible to combine a suffix array on $T$ with a lowest common ancestor data structure to allow longest common prefix queries in $O(1)$ time.  They further point out that all suffixes of $q$ can be added to the tree in $O(|q|)$ time.

Now, consider the case with $O(n/\tau)$ space for $\tau > 1$, as in Theorem~\ref{thm:result_small_space}.  
We split into two further cases.  

To find the longest common prefix between suffixes of $T$ and suffixes of $q$, 
we use the Monte-Carlo data structure of Bille et al.~\cite{BilleGoKn15}.  During preprocessing, we preprocess $T$ using their data structure in $O(n)$ time.  On each query $q$, we also preprocess $q$ using their data structure in $O(|q|)$ time.  
\iffull
In Appendix~\ref{sec:lcp_discussion}, 
\fi
we show how 
the data structure for $T$ and $q$ can be merged in $O(|q|)$ time, allowing us to answer longest common prefix queries between suffixes of the two strings---we also discuss why other longest common prefix data structures do not seem to work for our use case.
Each query requires $O(\tau + \log n)$ time and gives a correct answer with high probability.

To find the longest common prefix between two suffixes of $T$, we use the 
\iffull
deterministic 
\fi
data structure of Kosolobov and Sivukin~\cite{KosolobovSivukin24}, which requires $O(n/\tau)$ space and preprocessing, and $O(\tau)$ query time.  

These results immediately imply the following lemma.
 \begin{lemma}
    \label{lem:lcp}
   Let $\hat{s_1}$ be either a $(k+1)$-altered substring of $T$ or a $(k+1)$-altered substring of $q$, and let $\hat{s_2}$ be a $(k+1)$-altered substring of $T$.
   There is an $O(n/\tau)$ space data structure that can be built with $O(n)$ preprocessing on $T$ and $O(|q|)$ preprocessing on $q$ such that:
\begin{itemize}[topsep=0pt, noitemsep]
   \item If $\tau = 1$ or $s_1$ is from $T$, then  the data structure can find the first $k+1$ mismatches between $\hat{s_1}$ and $\hat{s_2}$ in $O(k\tau)$ time.
   \item If $\tau > 1$ and $s_1$ is from $q$, the data structure can find the first $k+1$ mismatches between $\hat{s_1}$ and $\hat{s_2}$ with high probability in $O(k\tau + k\log n)$ time.
   \end{itemize}
   \iffull
   From there, we can find $LCP(\hat{s_1}, \hat{s_2})$ and determine if $d_H(\hat{s_1},\hat{s_2}) \leq k$ in $O(1)$ additional time.
   \fi
\end{lemma}

With high probability means with probability $\geq 1 - 1/n^c$ for any constant $c$ (increasing $c$ increases the constant in the big-$O$ bounds).
During any query, the number of calls to Lemma~\ref{lem:lcp} can be loosely upper bounded by $O(n)$.  
Therefore, 
\iffull
since each call is correct with high probability, 
\fi
by a standard union bound, \emph{all} calls to Lemma~\ref{lem:lcp} during a given query are correct with probability $\geq 1 - 1/n^{c-1}$, i.e.\ with high probability after incrementing $c$.

\subsection{Dictionary Queries}
\label{sec:dictionary_version}
In this work, we phrase the problem as finding substrings of a text $T$ at Hamming distance $k$ from a query $q$.  However, many applications involve a slightly different problem: we are given a set of strings $s_1, s_2, \ldots, s_x$, and the goal is to find all strings $s_i$ with distance at most $k$ from $q$.  Following~\cite{ColeGoLe04}, we refer to this as the \defn{dictionary queries} variant of the problem, as opposed to the \defn{text indexing} variant used elsewhere in the paper.  In the dictionary queries variant, $n$ is the total length of all strings in the set.

It is well known that the dictionary queries variant reduces to the text indexing variant.  Assume that characters \texttt{1}, \texttt{2}, \ldots are not in $\Sigma$.  We interleave these characters with those of $q$ to create a new string: $q[1]~ \texttt{1}~ q[2]~ \texttt{2}~ q[3]~ \texttt{3}\ldots$ and so on.  We modify $s_1, \ldots, s_x$ the same way, e.g.\ $s_1$ becomes $s_1[1]~\texttt{1}~s_1[2]~\texttt{2}~s_1[3]~\texttt{3}\ldots$ and so on. We then concatenate the resulting strings to obtain $T$.  
The transformed $q$ has distance $\leq k$ from a substring of $T$ exactly when the original query has distance $\leq k$ from a string in $s_1, \ldots, s_x$.

\section{A Recursive Formulation of the CGL Tree}
\label{sec:cgl}

 In this section, we describe RecursiveCGL, our exposition of the CGL tree, and analyze it to prove Theorem~\ref{thm:result_improved_cgl}. 

Before formally defining our data structure, we give intuition behind the data structure, and state the key observations we use in its construction.

\paragraph{The Plan.}  To begin, we describe our strategy using an analogy with a static binary search tree (BST).  In a BST, each node stores a pivot $p$.  To ensure height $\log_2 n$, the pivot $p$ is the median of the node's descendants; the node has one child for the elements $< p$, and another for the elements $> p$.
A query traverses the tree recursively. At each node, $q$ is compared to $p$; the left node is traversed if $q < p$ and the right node is traversed if $q > p$.

We use a similar structure to allow us to find the suffixes of $T$ with Hamming distance at most $k$ from the query $q$.
 We define the data structure as a recursive tree, where each node stores a pivot $p$.

Now, the goal is to replicate a BST-style recursion for Hamming distance.  First, when constructing the data structure, we need to partition the strings $s$ in each node using their relationship with the pivot $p$---we will need to generalize beyond simply storing strings $< p$ and strings $> p$.   Second, when querying the data structure, our plan will be to compare $q$ to $p$, and use the result of this comparison to make recursive queries.

\paragraph{Pivot-Altering Strings.}
One way that our data structure differs from both the static BST and the classic CGL tree is that we recursively change the strings when building our tree, and we recursively change the query string itself during queries.  We define these changes using the pivot string $p$.

For any string $s$, we can \defn{pivot-alter} $s$ to match the pivot string $p$ by one more character: in particular, the \defn{pivot-altered} $s$ can be obtained by calculating $i = LCP(s,p)$, and by appending the alteration $(i+1, p[i+1])$ to $s$.  We will generally denote the result of this process $\hat{s}$ for a fixed $p$.

\paragraph{Recursive Idea of the Tree: Intuition.}
Recall that the goal of RecursiveCGL is to, for each tree node, compare $q$ to $p$, and recursively traverse the tree using the result of this comparison.
The following observation is the heart of this process in RecursiveCGL (and is essentially the main idea behind the CGL Tree~\cite{ColeGoLe04}).  
We first state our observation intuitively, then give a more formal statement.  

For any query $q$ and any string $s \neq q$, let $\ell$ be the first index where $q[\ell] \neq s[\ell]$.  If $p[\ell] = q[\ell]$, then altering $q$ makes it more similar to $s$: $d_H(\hat{q}, s) = d_H(q, s) - 1$.  Alternatively, if $p[\ell] = s[\ell]$, then altering $s$ makes it more similar to $q$: $d_H(q, \hat{s}) = d_H(q,s) - 1$.  If $p[\ell] \neq q[\ell]$ and $p[\ell] \neq s[\ell]$, then altering both $q$ and $s$ makes them more similar: $d_H(\hat{q}, \hat{s}) = d_H(q,s) - 1$.
We can rephrase these conditions using the LCP:  if $LCP(q,p) < LCP(s,p)$, then $p[\ell] = s[\ell]$ but $p[\ell] \neq q[\ell]$.

Let us briefly discuss why this observation is useful.  When creating the tree, we partition all strings into $O(1)$ sets according to their LCP with the pivot string $p$.  On a query $q$, we begin by calculating $LCP(q,p)$.  Let's look at two examples of how we recurse.
\begin{enumerate}[topsep=0pt, noitemsep]
\item if a set contains only strings $s$ with $LCP(s,p) > LCP(q,p)$, we can recursively query that set using $\hat{q}$, reducing the search radius by $1$.  
\item if a set contains only strings $s$ with $LCP(s,p) < LCP(q,p)$ we can recursively query all altered strings in the set using $q$, again reducing the search radius by $1$.  
\end{enumerate}
Since the search radius can only be decreased $k$ times, this process only searches a bounded number of nodes in the tree.

 \paragraph{Pivot-Altered String Structure.}
Now we state our observation more formally.  We begin with the case that $q$ is not a prefix of $p$.

 \begin{restatable}{observation}{correctnesscoreobs}
\label{obs:correctness_core}
Let $q$ and $s$ be arbitrary strings, and 
let $\hat{q}$ and $\hat{s}$ be the pivot-altered $q$ and $s$ respectively for a pivot string $p$. Let $i = LCP(q,p)$ and $j = LCP(s,p)$. 
Then if 
$i < |q|$ and $j < \min\{|s|, |p|\}$:
\footnote{The requirement that $j < \min\{|s|, |p|\}$ will always be satisfied: $s$ and $p$ will be $k$-altered suffixes of $T$; since $T$ ends with $2k+1$ copies of a special character \$, there must be some place where $s$ and $p$ differ.}
\begin{enumerate}[topsep=2pt, noitemsep]
	\item If $i < j$, then $d_H(\hat{q}, s) = d_H(q,s) - 1$;
	\item if $i > j$, then $d_H(q, \hat{s}) = d_H(q,s) - 1$;
	\item if $i = j$ and $q[i+1] \neq s[j+1]$, then $d_H(\hat{q}, \hat{s}) = d_H(q,s) - 1$.
\end{enumerate}
     
\end{restatable}

\begin{table}[t]
 \captionsetup{position=below}
 \begin{center}
    \renewcommand{\arraystretch}{1.5}
    \begin{tabular}{p{15ex} p{15ex} p{15ex}}
     $i < j$ & $i > j$ & $i = j$ \\
     \hline 
    $p$: \texttt{BANANA} & $p$: \texttt{BANANA} & $p$: \texttt{BANANA} \\
    $q$: \texttt{\hi{C}ANANC}\newline $s$: \texttt{BA\hi{C}ACA} & $q$: \texttt{BANA\hi{C}C} \newline $s$: \texttt{B\hi{C}NANA} & $q$: \texttt{BA\hi{C}CNA} \newline $s$: \texttt{BA\hi{D}ANA} \\
    $\hat{q}$: {\texttt{\hi{B}ANANC}} \newline $s$: \texttt{BA\hi{C}ACA} & $q$: \texttt{BANA\hi{C}C} \newline $\hat{s}$: \texttt{B\hi{A}NANA} & $\hat{q}$: \texttt{BA\hi{N}CNA} \newline $\hat{s}$: \texttt{BA\hi{N}ANA} 
    \end{tabular}
    \end{center}
\caption{An example for Observation~\ref{obs:correctness_core}. The first character mismatched between $q$ or $s$ and $p$---that is, the $i+1$st character of $q$ and $j+1$st character of $s$, which are altered in $\hat{q}$ and $\hat{s}$---is highlighted in red.}
\end{table}

If $i = |q|$, then the query is a prefix of the pivot string, and the structure changes
\iffull
---in particular, if $s$ first differs from the pivot string after $|q|$ characters, then $d_H(q,s) = 0$, so $s$ is one of the solution strings.  
\fi

 \begin{restatable}{observation}{correctnesscorematchobs}
\label{obs:correctness_core_match}
Let $q$ and $s$ be arbitrary strings, and 
let $\hat{q}$ and $\hat{s}$ be the pivot-altered $q$ and $s$ respectively for a pivot string $p$. Let $i = LCP(q,p)$ and $j = LCP(s,p)$.  Then if $i = |q|$  and $j < \min\{|s|, |p|\}$:
\begin{enumerate}[topsep=2pt,noitemsep]
	\item If $j \geq |q|$ then $d_H(q, s) = 0$;
	\item if $j < |q|$ then $d_H(q, \hat{s}) = d_H(q,s) - 1$.
\end{enumerate}
\end{restatable}
\iffull
We provide a proof of Observations~\ref{obs:correctness_core} and~\ref{obs:correctness_core_match} in Appendix~\ref{sec:omitted}.
\fi

\subsection{The Recursive CGL Tree}
\label{sec:recursive_data_structure}
 In this section we formally describe how to build a recursive data structure, RecursiveCGL, to solve the \problem{} problem.  
The goal of the data structure is to store all suffixes of $T$ to handle queries with distance bound at most $k$.
However, as mentioned above, as we recurse we alter the stored suffixes and change the search radius. 
 Therefore, we describe the data structure recursively as storing a set $S$ of $k$-altered suffixes of $T$, which can handle queries with distance bound at most $k'$ for some $k' \leq k$; we denote this structure RecursiveCGL$(S,k')$. 

To begin, build the First $k+1$ Mismatches Data Structure from Lemma~\ref{lem:lcp} on $T$ with $\tau = 1$.  (In Section~\ref{sec:space_efficient_cgl}, when we make the data structure space-efficient, we may use larger values of $\tau$.)

Our data structure is recursive.  
The initial recursive call is to RecursiveCGL$(S, k)$ where $S$ consists of all suffixes of $T$ (no alterations).  
In the base case 
 $|S| \leq 1$, RecursiveCGL$(S, k')$ stores $S$.

From now on we assume that $|S| > 1$ and describe how to build RecursiveCGL$(S,k')$ recursively.
First, we calculate the median of $S$ in lexicographic order; we call this the \defn{pivot string} $p$.

We will now use $p$ to partition the strings in $S$.
Our partition is defined with two goals 
in mind: (1) the partition should allow us to search for strings using Observation~\ref{obs:correctness_core}, and (2) the size of each partition should be at most $|S|/2$ to ensure tree height $\lceil \log_2 n\rceil$.

  Let $m$ be the median $LCP(s,p)$ for all strings $s \in S$. We partition $S\setminus \{p\}$ into the following four sets
  \iffull
, which we call the \defn{recursive subsets} of $S$.  
  \fi

\begin{definition}[Recursive Subsets]
\label{def:recursive_subsets}
For any set of strings $S$ with pivot $p$, where $m$ is the median value of $LCP(s,p)$ for $s\in S$,
\begin{itemize}[noitemsep,topsep=0pt]
	\item $S_{< m}$ consists of all strings $s \in S$ with $LCP(s,p) < m$;
	\item $S_{> m}$ consists of all strings $s \in S$ with $ LCP(s,p)> m$;
	\item $S_{< \ell}$ consists of all strings $s \in S$ with  $LCP(s,p) = m$ that occur before $p$ in lexicographic order; and 
	\item $S_{> \ell}$ consists of all strings $s \in S$ with $LCP(s,p) = m$ that occur after $p$ in lexicographic order.
\end{itemize}
\end{definition}
 \begin{table}[t]
 \captionsetup{position=below}
 \label{table:recursive_subsets}
 \begin{center}
     $S$ = 
     \texttt{AN\$} ,
     \texttt{BACDA\$},
      \texttt{BANAAAA\$} ,
      \texttt{BANANA\$},
      \texttt{BANAZAAAA\$},\\
      \texttt{BANCNAB\$} ,
      \texttt{BANCZZ\$} ,
     \texttt{CAT\$}
     
     \smallskip
 $p = \texttt{BANANA\$}$ \\
 
 \bigskip
    \renewcommand{\arraystretch}{1.2}
    \begin{tabular}{p{12ex} p{12ex} p{12ex} p{12ex}}
    \hline
     $S_{< m}$ & $S_{< \ell}$ & $S_{> \ell} $ & $S_{> m}$ \\
     \texttt{CAT\$} \newline \texttt{AN\$} \newline \texttt{BACDA\$} &  & \texttt{BANCNAB\$} \newline \texttt{BANCZZ\$} & \texttt{BANAAAA\$} \newline \texttt{BANAZAAAA\$}\\
     \hline
    \end{tabular}
    
    \bigskip
    \begin{tabular}{p{12ex} p{12ex} p{12ex} p{12ex}}
     \hline
     $\hat{S}_{< m}$ & $\hat{S}_{< \ell}$ & $\hat{S}_{> \ell} $ & $\hat{S}_{> m}$ \\
     \texttt{BAT\$} \newline \texttt{BN\$} \newline \texttt{BANDA\$} &  & \texttt{BANANAB\$} \newline \texttt{BANAZZ\$} & \texttt{BANANAA\$} \newline \texttt{BANANAAAA\$}\\
     \hline
    \end{tabular}
    \end{center}
\caption{An example of recursive subsets and altered subsets for $p = \texttt{BANANA\$}$.  In this example, $m = 3$ and  $S_{< \ell}$ is empty.  In reality, there are $2k+1$ copies of $\$$ at the end of each string, but these are truncated for simplicity.}
\end{table}
\begin{lemma}
\label{lem:set_size}
\iffull
All recursive subsets  $S_{< m}, S_{> m}, S_{< \ell}, S_{> \ell}$ have size at most $|S|/2$.
\fi
\end{lemma}
\iffull
\begin{proof}
There are at most $|S|/2$ strings that are before $p$ in lexicographic order; since $S_{< \ell}$ is a subset of these strings, it has size at most $|S|/2$.  The same argument applies for $S_{> \ell}$.
Similarly, since $m$ is the median $LCP$, $S_{< m}$ and $S_{> m}$ both have size at most $|S|/2$.
\end{proof}
\fi

For each of the four recursive subsets, we create an \defn{altered subset} by pivot-altering every string in the recursive set (recall that to ``pivot-alter'' an altered suffix $s$, we alter it so that it matches $p$ by one more character). 
Thus, 
$\hat{S}_{< m}$ stores the pivot-altered $s$ for all $s\in S_{< m}$; we similarly define $\hat{S}_{< \ell}, \hat{S}_{>\ell}$, and $\hat{S}_{> m}$.

We are ready to define our recursive data structure.
 \begin{itemize}[topsep=0pt,noitemsep]
	\item For each recursive subset, we build the data structure recursively with distance bound $k'$: we build RecursiveCGL$(S_{<m}, k')$, RecursiveCGL$(S_{< \ell}, k')$, RecursiveCGL$(S_{> \ell}, k')$, and 
    
    RecursiveCGL$(S_{> m}, k')$.
    
	\item If $k' > 0$, for each altered subset except $\hat{S}_{> m}$, we build the data structure recursively with distance bound $k'-1$: we build  
    RecursiveCGL$(\hat{S}_{<m}, k'-1)$, RecursiveCGL$(\hat{S}_{< \ell}, k' -1)$, 
	and RecursiveCGL$(\hat{S}_{> \ell}, k'-1)$.  We do not need to build RecursiveCGL$(\hat{S}_{> m},k'-1)$; see the discussion below.  
\end{itemize}

These recursive data structures form a tree structure: if $k' > 0$, RecursiveCGL$(S', k')$ consists of the pivot string, and $7$ pointers to the recursive data structures defined above (four recursive subsets, and three altered subsets).    
The  pivot string consists of a pointer to a suffix of $T$ and at most $k$ alterations, which requires $O(k)$ space.  Combining with the space to store the $O(1)$ pointers, we can store RecursiveCGL$(S',k')$ in $O(k)$ space.

 If $k' = 0$, we recurse on the recursive subsets, but we omit the altered sets; thus RecursiveCGL$(S', 0)$ consists of the pivot string, and $4$ pointers to recursive data structures on the recursive subsets.  Again, we can store the pivot using a pointer and the list of operations, giving $O(k)$ space per node.

\paragraph{The Recursive Data Structure as a Tree.}
 For the remainder of this paper, we will treat this data structure as a degree $7$ tree, which we refer to as%
\footnote{The subscript $0$ is to differentiate it from the space-efficient versions of this data structure we will look at later.}
 $\mathcal{T}_0$.  
The root of $\mathcal{T}_0$ represents RecursiveCGL$(S, k)$, where $S$ consists of all suffixes of $T$, each with an empty set of alterations; the children of each node represent the recursive data structures on the recursive subsets and altered subsets. 
This is essentially the same idea as a standard binary search tree---each child of a node is the root of a recursive tree on a subset of the data.
We refer to a node of a tree, and the recursive data structure it represents, interchangeably.  
The height of $\calT_0$ is at most $\lceil \log n\rceil$ by Lemma~\ref{lem:set_size}.  

We say that a node RecursiveCGL$(S,k)$ \defn{contains} a suffix $s$ of $T$ if $(s,a)\in S$ for some (possibly empty) set of alterations $a$.  We define the \defn{size} of a node to be $|S|$.

\subsubsection{Queries}
\label{sec:cgl_queries}
 We describe how to perform a query $q$ with search radius $r\leq k$.
The query is recursive, traversing the tree $\mathcal{T}_0$; as the tree is traversed, the query is altered 
and the search radius decreases.  In Section~\ref{sec:cgl_queries}, we use the notation $q$ solely to refer to the original query string, and $q'$ to a recursive (possibly altered) query; similarly, $r$ is the original search radius, and $r'$ is the recursive search radius.
The query begins at the root of $\mathcal{T}_0$. 

We now describe how a query is performed.  We split into two cases based on if the recursive search radius $r' > 0$ or $r' = 0$.  For the case where $r' > 0$, we split into two further cases based on whether or not $q'$ is a prefix of the current node.  For the case where $r'=0$, we give two strategies.  First, we explain how to traverse the tree recursively.  Then, we summarize how storing additional data structures allows us to traverse the tree much more quickly;
\iffull
we fully explain this strategy and bound its performance in 
Section~\ref{sec:improving_logs}.
\fi

\paragraph{Positive Search Radius $r' > 0$.}
We describe how to traverse a node $v$ representing RecursiveCGL$(S, k')$ for some $S$; let $r'$ be the current search radius.  
First, we determine whether the distance from $q'$ to the pivot string $p$ of $v$ is $\leq r'$ in $O(k)$ time using Lemma~\ref{lem:lcp}. If so, we return $p$.
In the base case, if $|S| \leq 1$, we are done.

Otherwise, we find $i \coloneq LCP(q',p)$ using Lemma~\ref{lem:lcp}.
Once we find $i$, we can 
pivot-alter $q'$ in $O(1)$ time; denote the result  $\hat{q}$.  We can also determine if $q'$ is before or after $p$ in lexicographic order in $O(1)$ time.

We split into two subcases based on if $i < |q'|$ or $i = |q'|$.

If $i < |q'|$, we split into four cases based on $i$ and whether $q'$ is before or after $p$ in lexicographic order.  In all cases, we perform at most $4$ recursive queries: at most $1$ recursive query with search radius $r$, and at most $3$ with search radius $r-1$.  
These recursive queries can be obtained immediately from Observation~\ref{obs:correctness_core}; 
we list them explicitly in Table~\ref{table:queries}.

\begin{table*}
\centering
 \captionsetup{position=below}
\begin{tabular}{c|c c|c c|c c|c c}
    &  $S_{< m}$ & $\hat{S}_{< m}$  &  $S_{< \ell}$ & $\hat{S}_{< \ell}$  &  $S_{> \ell}$ & $\hat{S}_{> \ell}$  &  $S_{> m}$  & $\hat{S}_{> m}$ \\ \hline
    $i< m$ & $q',r$ & & $\hat{q}, r-1$ & & $\hat{q}, r-1$ & & $\hat{q}, r-1$ \\
     $i=m$, $q' < p$ & & $q',r-1$ & ${q'}, r$ & & & $\hat{q}, r-1$ & $\hat{q}, r-1$ \\
     $i=m$, $q' > p$ & & $q',r-1$ & &$\hat{q}, r-1$  & ${q'}, r$ &  & $\hat{q}, r-1$ \\
    $i> m$ & & $q',r-1$ & &${q'}, r-1$  & & ${q'}, r-1$ & ${q'}, r$ \\
\end{tabular}
\caption{The recursive calls made based on the query if $i < |q'|$ using Observation~\ref{obs:correctness_core}.  We use $q' < p$ to denote that $q'$ is before $p$ in lexicographic order.}
\label{table:queries}

\bigskip
\begin{tabular}{c|c c|c c|c c|c }
    & $S_{< m}$ & $\hat{S}_{< m}$  &  $S_{< \ell}$ & $\hat{S}_{< \ell}$  &  $S_{> \ell}$ & $\hat{S}_{> \ell}$  &  $S_{> m}$ \\ \hline
    $i< m$ & $q',r$ & & DFS & & DFS & & DFS \\
    $i= m$ & & $q',r-1$ & & ${q'}, r-1$ & & ${q'}, r-1$ & DFS \\
    $i> m$ & & $q',r-1$ & &${q'}, r-1$  & & ${q'}, r-1$ & ${q'}, r$ \\
\end{tabular}
\caption{The recursive calls made based on the query if $i = |q'|$ using Observation~\ref{obs:correctness_core_match}.}
\label{table:match_queries}
\end{table*}


If $i = |q'|$, then $q'$ is a prefix of the pivot string $p$.  This changes our recursive calls since \emph{any} string $s\in S$ with $LCP(s,p) > i$ has $d_H(q',s) = 0$.  
We handle these items by performing a depth-first search (DFS) on the tree to recover $S$: we recursively search all unaltered children of each node; the pivots of these nodes are exactly the elements of $S$.
(We chose DFS arbitrarily; this could be handled equally well by BFS or any other kind of tree traversal.)

\iffull
\paragraph{Query Method Discussion.}
To briefly justify these queries, consider the case $i = m$ and $q' > p$ (the third row of Table~\ref{table:queries}).  We describe this case in detail:
we describe each recursive call one by one, with a brief justification for each.   
The other cases are similar.

First, we recurse in $S_{< m}$.  Any $s\in S_{< m}$ satisifies $LCP(s,p) < LCP(q,p)$. Then by Observation~\ref{obs:correctness_core}, any $s\in S_{< m}$ has $d_H(s,q) \leq r$ if and only if $d_H(\hat{s}, q') \leq r-1$.  Therefore, our first recursive query is to RecursiveCGL$(\hat{S}_{< m}, k-1)$ with query $q'$ and search radius $r-1$.

Then, we recurse on $S_{> m}$.  Any $s\in S_{> m}$ satisfies $LCP(s,p) > LCP(q',p)$. Then by Observation~\ref{obs:correctness_core}, any $s\in S_{> m}$ has $d_H(s,q') \leq r$ if and only if $d_H(s, \hat{q}) \leq r-1$.  Therefore, our second recursive query is to RecursiveCGL$({S}_{> m}, k-1)$ with query $\hat{q}$ and search radius $r-1$.

Next, we recurse on $S_{< \ell}$.  Any string $s\in S_{< \ell}$ must have $LCP(s,p) = LCP(q',p)$ but (since $q' > p$) $s[m] \neq q'[m]$.
By Observation~\ref{obs:correctness_core}, $d_H(s,q') \leq r$ if and only if $d_H(\hat{s}, \hat{q})\leq r-1$.
Thus, we make 
a recursive query with $\hat{q}$ and search radius $r-1$ to 
RecursiveCGL$(\hat{S}_{< \ell}, k-1)$.

Finally, we make a recursive query with $q'$ and search radius $r$ to RecursiveCGL$(S_{> \ell}, k)$.  

\fi

\paragraph{Traversing the Tree with $r'=0$.}
 If $r'=0$, we can simply ignore all recursive calls with search radius $r-1$.  
There is only one non-DFS recursion for each node
\iffull
; we can determine which applies using Lemma~\ref{lem:lcp}.
\fi

For the remainder of this paragraph, we define traversing the nodes in more detail, giving some structural definitions that will be useful throughout the rest of the paper.  Then, we give a corollary that significantly speeds up our algorithm: rather than traversing the tree one node at a time for $r' = 0$, we can quickly ``jump'' to the final result of the traversal.

Let us look at the recursive calls with search radius $r$.  If $i < m$ we recurse in $S_{< m}$, if $i = m$ and $\hat{q} < p$ we recurse on $S_{< \ell}$, and so on.  In general, we traverse to the node whose recursive set would contain the query $\hat{q}$.

This traversal process is currently most important for ${q'}$, but in fact it is well-defined for any string $s$: starting at any node $v\in \calT_0$, we traverse to the descendant of $v$ whose recursive set would contain $s$, continuing until reaching a leaf.  
We call this process a \defn{manual traversal}; we call the last node reached in a manual traversal
the \defn{traversal destination}.  
\iffull
We define this process for any string $s$, and any subtree of $\calT_0$, as we will need these more  general definitions in Sections~\ref{sec:space_efficient_cgl} and~\ref{sec:improving_logs}.
\fi

\begin{definition}[Traversal Destination; Manual Traversal]
\label{def:traversal_destination}
Let $\calT'$ be a subtree of $\calT_0$ and let $s$ be an arbitrary string.
We define a path through $\calT'$ 
which we call the \defn{manual traversal}.  
We define one recursive subset to recurse on.
If $LCP(s,p) < m$ we recurse on $S_{< m}$; if $LCP(s,p) > m$ we recurse on $S_{> m}$;  if $LCP(s,p) = m$ and $s < p$ we recurse on $S_{< \ell}$; if $LCP(s,p) = m$ and $s > p$ we recurse on $S_{> \ell}$.
We define the \defn{traversal destination} to be the leaf of $\calT'$ reached by this process.  
\end{definition}

Our goal is to avoid the manual traversal, which may traverse $\Omega(\log n)$ nodes.  In 
Section~\ref{sec:improving_logs} 
we show that we can indeed find the traversal destination in $\text{poly}(k, \log\log n)$ time.

 However, this result is not sufficient to answer a query---we cannot just jump to the final node traversed, as we may skip matching strings found along the way.  
Specifically, as we traverse the tree, we compare the query $q$ to the pivot of each node along the traversal path.  If they match---$q$ is a prefix of $p$---then we must return $p$ as a solution to the problem.  Then, using Table~\ref{table:match_queries}, we may call one or more DFS subroutines to return further matching pivots.

The following definition address this gap: we define the set of nodes which we must traverse to correctly answer a query with radius $r' = 0$.  

\begin{definition}
Let $\calT'$ be a subtree of $\calT_0$, let $v$ be a node in $\calT'$, and let $q'$ be a $k$-altered query.  The \defn{Matching Nodes for $v$ and $q'$}, denoted $P(v,q', \calT')$, are the set of nodes $v'\in \calT'$ such that:
\iffull
\begin{enumerate}[noitemsep,topsep=1pt]
\item  $v'$ is a descendant of $v$ in $\calT'$, and 
\item $q'$ is a prefix of $p$.
\end{enumerate}
\fi
\end{definition}

 The following corollary shows that we can find the Matching Nodes of a query quickly, without manually traversing the tree.
\iffull
This corollary is a special case 
of Lemma~\ref{lem:fast_query}, stated and proven in Section~\ref{sec:improving_logs}.
\fi

\begin{corollary}
\label{cor:fast_query_original_cgl}
There is a data structure using $O(n\Smab)$ space that, for any $k$-altered query $q'$ can find $P(v,q', \calT_0)$ in time
 \[
 O\left(k^2 + k\log\log n + k|P(v,q', \calT_0)|\right).
 \]
\end{corollary}

With this data structure, we can finish our queries.  To query a node $v$ of the tree with query $q'$ and search radius $r' = 0$, we use Corollary~\ref{cor:fast_query_original_cgl} to find $P(v,q', \calT_0)$.  
Each pivot $(s_i, a)$ of a node in $P(v,q', \calT_0)$ satisfies $d_H(s_i, q)\leq r$; we return all such $s_i$.

\subsection{Analysis}
\label{sec:analysis}
 We analyze the query time and space of RecursiveCGL to prove Theorem~\ref{thm:result_improved_cgl}.  Both can be analyzed by substituting into the recurrence implied by the data structure's definition (as was done in~\cite{ColeGoLe04}).  However, it will be helpful for our later analysis to have additional structure. 
  In particular, we give an analysis which labels the edges of $\calT_0$; then we use these labels to obtain our bounds.

\subsubsection{Query Time.}

We begin by bounding the number of times we invoke Corollary~\ref{cor:fast_query_original_cgl}.   This lemma forms the core of our running time proof.
Let $\calT_q$ be the subtree of $\calT_0$ traversed using recursive calls that (1) have radius $r > 0$, and (2) are not a DFS call; thus we invoke Corollary~\ref{cor:fast_query_original_cgl} once for each leaf of $\calT_q$.

\begin{lemma}
    \label{lem:fast_query_calls}
     $|\calT_q|  = O\left(3^k\binom{\log n}{k}\right)$.
\end{lemma}
\iffull
\begin{proof}
Let us define a \defn{query label}; we assign the query labels to the edges in $\calT_q$.  
Consider a node $v\in \calT_q$ that $q$ searches with radius $r'$.  
Then $q$ searches at most one child of $v$ with radius $r'$; we label the edge to this child \texttt{s}.
Any children searched by $q$ with search radius $r'-1$ have edges labelled \texttt{1}, \texttt{2}, \texttt{3} from left to right.  
(Recall that by definition $v$ has at most four children, and at most one has search radius $r'$---see Table~\ref{table:queries}.)

For any node $v$ of $\calT_q$, let $\ell$ be the string obtained by concatenating the labels of all edges on the path from the root to the $v$; call this the \defn{query label} of $v$.  Since all edges from a node have distinct labels, $\ell$ is unique for each node. 

Let us look at the properties of any path label $\ell$.  Since there are $\leq \lceil \log_2 n\rceil $ edges on the path from the root to any node in $\calT_0$ by Lemma~\ref{lem:set_size}, $\ell$ has length at most $\lceil \log_2 n\rceil$.  
The number of occurrences of \texttt{1}, \texttt{2}, or \texttt{3} in $\ell$ is at most $k$, since $k'$ starts at $k$, is never less than $0$, and is decremented every time one of these characters is added.  

 We bound the size of $\calT_q$
 by counting the number of distinct path labels.  
 
Consider a path $\lambda$ with length $i$, and $j$ characters that are \texttt{1}, \texttt{2}, or \texttt{3}.  There are $3^j$ ways to choose the values of \texttt{1}, \texttt{2}, or \texttt{3}, and $\binom{i}{j}$ ways to place these characters in $\lambda$.
    We can upper bound that there are at most 
      \[
\sum_{i = 0}^{\lceil \log_2 n\rceil } \sum_{j = 0}^{k-1} 3^j\binom{i}{j} \leq 
\sum_{j = 0}^{k-1} \sum_{i = j}^{\lceil \log_2 n\rceil } 3^j \binom{i}{j}
\leq 
\sum_{j = 0}^{k-1}  3^j \binom{\log n + 1}{j + 1}
\]
path labels.
This last step is well-known; see e.g.~\cite[Identity 5.10]{GrahamKnPa94}.

 Then $\binom{\log n}{i} \leq \binom{\log n}{i+1}$ for $i < k$, so (recall that we assume $k < (\log n)/2$):
\[
|\calT_q|\leq \sum_{j = 0}^{k-1}  3^j \binom{\log n + 1}{j + 1} \leq 
\sum_{j = 0}^{k-1}  3^j \binom{\log n + 1}{k} \leq 
3^{k}\binom{\log n + 1}{k}.
\]
We have that $\binom{\log n + 1}{k} = \binom{\log n}{k}\cdot\frac{\log n + 1}{\log n + 1 - k} \leq 2\binom{\log n}{k}$ since $k < (\log n)/2$.
\end{proof}
\fi

To begin bounding the other terms, 
we show that each string is only considered once by the query---this is crucial for the query time to be linear in $(\occ)$.
\begin{lemma}
\label{lem:no_duplicates}
    Let $s_i$ be a suffix of $T$, and let $Q$ be the nodes of $\calT_0$ traversed during a query.  Then there is at most $1$ node in $Q$ with pivot $(s_i, a)$ for some set of alterations $a$.
\end{lemma}
\iffull
\begin{proof}
Consider a node $v$ with set $S$ with $(s_i, a)\in S$ for some $a$.  If $(s_i, a)$ is not the pivot of $S$, then $(s_i, a)$ is in exactly one recursive subset $S'$, and $(s_i, a')$ is in exactly one altered subset $\hat{S'}$.  By Tables~\ref{table:queries} and~\ref{table:match_queries}, at most one of $S'$ and $\hat{S'}$ is traversed by $q$.  If $(s_i,a)$ is the pivot, then it is not stored recursively.
\end{proof}
\fi

\begin{lemma}
\label{lem:cgl_query_time}
    The total time to perform a query $q$ with output size $\occ$ is
    \[
    O\left(|q| +  3^k\binom{\log n}{k} (k^2  + k\log\log n) + k\cdot \occ \right).
    \]
\end{lemma}
\iffull
\begin{proof}
We must perform $O(|q|)$ time processing $q$ for Lemma~\ref{lem:lcp} and Corollary~\ref{cor:fast_query_original_cgl}.  This processing is done once per query.

Each node of $\calT_q$ can be traversed in $O(k)$ time using Lemma~\ref{lem:lcp}. 
Corollary~\ref{cor:fast_query_original_cgl} is called at most once per node of $\calT_q$; each call takes
$O(k^2 + k\log\log n + k|P(v, q', \calT_0)|)$
time. 
Putting the $|P(\cdot,\cdot)|$ term aside for a moment and 
combining with Lemma~\ref{lem:fast_query_calls} gives 
\[
O(3^k\binom{\log n}{k} (k^2 + k\log\log n))
\]
time.

 Now, let us consider all remaining costs: the $|P(\cdot, \cdot)|$ terms in all calls to Corollary~\ref{cor:fast_query_original_cgl}, and the cost of all calls to DFS.  Each pivot of a node in $P(v,q', \calT_0)$ for any $v$ and $q'$ is returned as a solution string, as is any pivot in any DFS call.  By Lemma~\ref{lem:no_duplicates}, each solution string only occurs once; thus, the total sum of these costs is $O(k\cdot \occ)$.
 \end{proof}
 \fi

\subsubsection{Space.}
To bound the space, we assign labels to each edge of the tree, and bound the number of possible labels combinatorially.
For now, this is an analysis tool. However, this definition will be important for our space-efficient approach in Section~\ref{sec:space_efficient_cgl}, where we will maintain path labels explicitly.
\begin{definition} 
\label{def:path_label}
The \defn{path label} of any node $v\in \calT_0$ is a string consisting of characters \texttt{a} and \texttt{u}.  It is defined recursively as follows. The path label of the root is the empty string.
Let $\ell$ be the path label of a node $v$.  For any child $c$ of $v$, the path label of $c$ can be obtained by adding \texttt{u} to $\ell$ if $c$ is an (unaltered) recursive subset, and by adding \texttt{a} to $\ell$ if $c$ is an altered subset.
\end{definition}

\begin{observation}
\label{obs:path_label_basics}
The length of a path label is at most $\lceil \log_2 n\rceil$, and
each path label contains at most $k$ \textup{\texttt{a}}'s.
\end{observation}

The path labels are useful for bounding the space (and will be useful for defining the functions in Section~\ref{sec:space_efficient_cgl}) since they are a way to track, for each suffix $s_i$ of $T$, where altered suffixes $(s_i, a)$ are stored in $\calT_0$.  
\iffull
We prove the following observation in the appendix; the proof is essentially the same as that of Lemma~\ref{lem:no_duplicates}.
\fi

 \begin{restatable}{observation}{pathdistinctlabels}
\label{obs:path_distinct_labels}
    For some suffix $s_i$ of $T$, consider two nodes $v_1$ and $v_2$ in $\calT_0$ (and alterations $a_1$ and $a_2$) such that $(s_i,a_1)$ is a pivot of $v_1$ and $(s_i, a_2)$ is a pivot of $v_2$.  Then the path labels of $v_1$ and $v_2$ are distinct; furthermore, the path label of one is not a prefix of the path label of the other.
\end{restatable}

With this observation in hand 
\iffull
we can prove our space bound.
\fi

\begin{lemma}
\label{lem:space}
Any suffix $s_i$ of $T$ is the pivot of $O(k\Smab)$ nodes of $\calT_0$; therefore,
RecursiveCGL requires $O(nk^2\Smab)$ total words of space.
\end{lemma}
\iffull
\begin{proof}
    First, let's bound the space to store $\calT_0$.  
    We charge the space required by each node to its pivot---any node built on an empty set must have a nonempty parent, so the number of empty nodes is at most $7$ times the number of nonempty nodes.
    
    By Observation~\ref{obs:path_distinct_labels}, for each suffix $s_i$ of $T$, the path label of all nodes that have $(s_i, a)$ as a pivot for some $a$ are distinct, and none is the prefix of another.  
    Consider the set of all path labels of nodes that have $s_i$ as a pivot.  Without loss of generality all have length $\log n$ (we ignore the ceiling for simplicity as it does not affect the analysis)---otherwise, we could add characters to one of them while retaining the number of path labels since no label is a prefix of another.  
    Then the number of path labels satisfying Observation~\ref{obs:path_label_basics} is at most (recalling that we assume $k < (\log n)/2$)
    \[
    \sum_{j = 0}^k \binom{\log n}{j} = O\left(k\binom{\log n}{k}\right).
    \]

    Summing over all $s_i$, there are at most $O(nk\binom{\log n}{k})$ nodes in $T_0$.  Each node requires $O(k)$ space ($O(k)$ space to store the pivot since it is an altered suffix of $T$, plus $O(1)$ pointers to further recursive data structures),
    for $O(nk^2\binom{\log n}{k})$ space overall.

    The LCP data structure from Lemma~\ref{lem:lcp} and the data structure for $r'=0$ from Corollary~\ref{cor:fast_query_original_cgl} both require $O(n\binom{\log n}{k})$ space and do not increase this bound.
\end{proof}
\fi

\subsubsection{Preprocessing Time.}

 The data structures for Lemma~\ref{lem:lcp} and Corollary~\ref{cor:fast_query_original_cgl} can be built in $O(nk^2\Smab\log n)$ time.  All that remains to prove Theorem~\ref{thm:result_improved_cgl} is to bound the time to build $\calT_0$.  

\iffull
In this section, we have $\tau = 1$;  
 in the following lemma we also consider the $\tau > 1$ case so we can use it in Section~\ref{sec:analysis}.
 \fi

 \begin{lemma}
 \label{lem:cgl_preprocessing}
    $\calT_0$ can be built in time 
    $O(nk^2 \tau \Smab \log n)$.
\end{lemma}
\iffull
\begin{proof}
  First, if $\tau = 1$, we show that a node with a set $S$ can be processed in time $O(k|S|)$. 
   Using Lemma~\ref{lem:lcp}, we can find the pivot in $O(k|S|)$ time using the standard median-finding algorithm; then, we can find the 
   LCP of all strings in $O(k|S|)$ time, and we can find the median LCP in $O(|S|)$ time.    Then, we can find the recursive subset of all strings using Lemma~\ref{lem:lcp} in $O(k|S|)$ time and recurse.
 If $\tau > 1$, the time to process each set using Lemma~\ref{lem:lcp} is $O(k\tau |S|))$.
   
   Since each suffix $s_i$ is the pivot of at most $k\binom{\log n}{k}$ nodes by Lemma~\ref{lem:space}, 
   $s_i \in S$ only if $s_i$ is a pivot of the descendant of the node with set $S$, and $\calT_0$ has height $O(\log n)$,
   each suffix $s_i$ contributes to $O(k\Smab\log n)$ sets, giving the lemma.
\end{proof}
\fi

\section{The Space-Efficient Data Structure}
\label{sec:space_efficient_cgl}

In this section we describe the data structure satisfying 
Theorems~\ref{thm:result_large_space} and~\ref{thm:result_small_space}.  We first describe how to truncate the tree to save space, then we describe how function inversion can help us perform queries on the truncated tree. Then we describe how to perform function inversion efficiently, proving  Theorem~\ref{thm:function_inversion_all}, and finally we put it all together to achieve the final bounds.

\subsection{The Truncated Tree}
\label{sec:truncating_cgl}

We define the \defn{truncated tree} similarly to RecursiveCGL, but with a key difference in the base case.  
If $|S| \leq \sigma$, then the node consists only of an integer \defn{label}; nothing else is stored, and we do not recurse further.  
 We use $\mathcal{T}$ to denote the truncated tree constructed on all suffixes of $T$. 
 
Since the nodes of $\mathcal{T}$ are also nodes of $\calT_0$, we can define path labels 
using Definition~\ref{def:path_label}
; we will use path labels both for space analysis and to define the functions for function inversion.

Consider all leaves in $\mathcal{T}$ with a given path label $\lambda$. These leaves can be found by traversing $\calT$ recursively.  We label all leaves in $\mathcal{T}$ with path label $\lambda$ using increasing integers starting at $1$.  The following lemma implies that the maximum label of any node is $O(n/\sigma)$.

\begin{lemma}
\label{lem:truncated_num_nodes}
For any path label $\lambda$, $\mathcal{T}$ has $O(n/\sigma)$ leaves with path label $\lambda$. $\mathcal{T}$ has 
$O(n k \Smab/\sigma)$ nodes in total.
\end{lemma}
\iffull
\begin{proof}
Let $\mathcal{T}_l$ consist of the internal nodes of $\mathcal{T}$. By definition, these are exactly the nodes of $\calT_0$ built on a set $S$ with $|S| > \sigma$.  Since each node has at most $7$ children, we can bound the number of nodes in $\mathcal{T}$ up to constants by bounding the number of nodes in $\mathcal{T}_l$.

Consider the leaves of $\mathcal{T}_l$.  
By Lemma~\ref{lem:space}, any suffix $s$ of $T$ satisfies that $(s,a)$ for 
some $a$ is the pivot of $O(k\binom{\log n}{k})$ leaves of $\calT_0$.  
A leaf of $\calT_l$ contains $s$ only if $(s,a)$ is the pivot of a descendant of the leaf in $\calT_0$; thus, the sum of the sizes of the leaves of $\mathcal{T}_l$ is at most $O(nk\binom{\log n}{k})$.   
Furthermore, by Observation~\ref{obs:path_distinct_labels}, $s$ can only contribute to one leaf of $\calT_l$ with label $\lambda$.
Each has size at least $\sigma$; thus there are $O(n/\sigma)$ leaves in $\calT_l$ with label $\lambda$, and $O(n\binom{\log n}{k}/\sigma)$ leaves in $\calT_l$.

Let $\mathcal{T}_d$ be the nodes of $\mathcal{T}_l$ whose ancestors all have at least $2$ children in $\mathcal{T}_l$.  Immediately, $\mathcal{T}_d$ has $O(n\binom{\log n}{k}/\sigma)$ nodes; we are left to bound the size of $\mathcal{T}_l\setminus\mathcal{T}_d$: the number of nodes with an ancestor of degree $1$.

Consider a node $v$ in $\mathcal{T}_l$ with degree $1$ built on a set $S_v$.  At most one child of $v$ has size more than $\sigma$; it has size at most $|S_v|/2$ by Lemma~\ref{lem:set_size}.  The remaining children of $v$ have size at most $\sigma$.  Since the size of a node is at most the sum of the sizes of its children, we have that $|S_v| \leq |S_v|/2 + 6\sigma$. 
Solving, $|S_v| \leq 12\sigma$.

Therefore, by Lemma~\ref{lem:set_size}, all children of $v$ have size at most $6\sigma$; extending we obtain that all children at depth $4$ in the subtree rooted at $v$ have size $\leq \sigma$.  Since the tree has degree $7$, the subtree rooted at $v$ has $7^4$ nodes in $\calT$.

Therefore, if $v'\in \mathcal{T}_l$, each child of $v'$ that is not in $\mathcal{T}_l$ can have at most $7^4$ descendants.  
Therefore, $v'$ can have in total at most $7^5 = O(1)$ descendants that are not in $\mathcal{T}_d$. 
If two nodes in $\calT$ have the same path label, they must have an ancestor in $\calT_l$ with the same path label.  
Thus, we obtain the lemma.

We point out that while this proof has large constants, they can likely be avoided.  First, the above analysis can likely be tightened.  Second, rather than incurring the large $7^5$ constant, we can trim all nodes in $\mathcal{T}_l\setminus\mathcal{T}_d$ from $\mathcal{T}$.  This increases the largest size of a node to $12\sigma$, increasing the query time in our results by a factor $12$, but removing this large constant in the space bound. 
\end{proof}
\fi

 \subsection{Defining the Functions We Will Invert}
\label{sec:functions}

We define a function $f_{\lambda}$ for each path label $\lambda$.  For any $i\in [n]$, $f_{\lambda}(i)$ is the label of the leaf $\ell \in \mathcal{T}$ such that $\ell$ has path label $\lambda$ and $\ell$ contains $s_i$.  
There is at most one such $\ell$ by 
Observation~\ref{obs:path_distinct_labels} since each leaf $\ell$ of $\calT$ contains exactly the suffixes that are pivots of the descendants of $\ell$ in $\calT_0$;
if no such $\ell$ exists then we define $f_{\lambda}(i) = \bot$.  

Thus, $f_{\lambda}(i): [n] \rightarrow [O(n/\sigma)] \cup \{\bot\}$.  
In the following discussion, we will pad the codomain of $f_{\lambda}(i)$ to have size exactly $n$; in other words, $f_{\lambda}(i): [n] \rightarrow [n-1] \cup \{\bot\}$.

From this definition, we immediately observe that all elements of the codomain except $\bot$ have a small preimage:  since every leaf in $\calT$ contains $O(\sigma)$ elements, we have $|f_{\lambda}^{-1}(j)| = O(\sigma)$ for all $j\in [n-1]$.  This observation will be important in the subsequent analysis.

\paragraph{Evaluating the Functions.}
Let $E(f_{\lambda})$ be the worst-case time to evaluate $f_{\lambda}(i)$ for any $\lambda, i$.
We can evaluate $f_{\lambda}(i)$ in $O(k\tau\log n)$ time by traversing $\mathcal{T}$ (the height of the tree is $O(\log n)$; using Lemma~\ref{lem:lcp} we can determine which subset contains $s_i$ in $O(k\tau)$ time since $s_i$ and the pivot $p$ are altered suffixes of $T$).  

However, the key idea behind function inversion~\cite{FiatNaor00} is that we repeatedly calculate $f_{\lambda}(\cdot)$ (in the forward direction) in order to find an element's inverse $f_{\lambda}^{-1}(j)$.  Therefore, the time spent calculating $f_{\lambda}(\cdot)$ has a significant impact on our bounds---we would like to avoid spending $\Omega(\log n)$ time traversing $\calT$.  We point out that this very similar to the goal of Corollary~\ref{cor:fast_query_original_cgl}: we want to ``jump'' right to the result of the traversal.

 In Section~\ref{sec:improving_logs}, 
 we improve the time to evaluate $f_{\lambda}(\cdot)$ to $O(k^3 + k^2\log\log n)$
 ; see Lemma~\ref{lem:fast_function}.
The space bound of 
Lemma~\ref{lem:fast_function} 
has an additive $\Theta(n)$ term; therefore, we only use this lemma if $\tau = 1$.  
Our succinct data structure ($\tau > 1$) evaluates $f_{\lambda}(i)$ by traversing $\calT$
in $O(k\tau\log n)$ time.

\subsection{Proof of Theorem~\ref{thm:function_inversion_all}}
In this section, we prove Theorem~\ref{thm:function_inversion_all}, giving a data structure to efficiently invert an arbitrary function $f$ satisfying the requirements of Theorem~\ref{thm:function_inversion_all}.  Later, we will use this result with $f = f_{\lambda}$ for the functions $f_{\lambda}$ defined above to achieve our final bounds.

This proof is in two parts: first, we describe how the classic Fiat-Naor data structure extends to our setting; then, we extend their analysis to find all points in the inverse to Theorem~\ref{thm:function_inversion_all}.

\subsubsection{Recalling the Fiat-Naor Data Structure}
\label{sec:fiat_naor_review}
 Let's begin with a review of the data structure from~\cite{FiatNaor00}.  For the most part, our exposition here is identical to theirs; however, we tweak their parameters slightly to fit our analysis, and we omit the data structure they use to store high-indegree elements.

We briefly describe the idea behind their data structure before giving a detailed exposition.  Consider an element $i\in [n]$; apply $f$ repeatedly to $i$ to obtain $f(i), f(f(i))$, and so on for $\sigma$ operations; call this a ``chain.''  If we store a mapping from the last element to the first element in the chain, we can invert all elements in the chain as follows. To invert an element $j$, we start applying $f$ to $j$ until reaching the last element in the chain; following the pointer to $i$, we continue applying $f$ until reaching $j$ again.  The penultimate element in this strategy is in $f^{-1}(j)$.

 The challenge is ensuring that all elements are in some chain, and ensuring this works for any $f$.  This is done by composing $f$ with a random hash function, and creating a large number of chains (called a ``cluster'') for the composed function.  Repeatedly creating clusters, with a new hash each time, is sufficient to lower bound the probability that $j$ is in some chain.

\paragraph{Creating the Chains.}
We create a sequence of \defn{clusters}.  For each cluster $c$, we select a $2\sigma$-independent hash function
$g_c: [2n]\rightarrow [n]$ from the 
family given 
in~\cite[Section 4.4]{FiatNaor00}.  This hash family can be stored in $O(1)$ amortized space, and can be evaluated in $O(\log\sigma)$ amortized time.  (The amortization is over all hashes for all clusters; since each query is always repeated for each cluster, we can treat these bounds as worst-case on a given cluster.)
We define a function $h_c$ that composes $g_c$ and $f$, avoiding the high-indegree element $\bot$:
\[
h_c(i) = 
\begin{cases}
g_c(f(i)) & \text{ if } f(i) \neq \bot; \\
g_c(n + i) & \text{ otherwise.}
\end{cases}
\]

This function differs from the one in~\cite{FiatNaor00}: they avoid high-indegree elements by repeatedly hashing to avoid elements sampled in a lookup table. Since we know that $\bot$ is our only high-indegree element we can instead handle it explicitly.  This hash function is defined so that $\Pr_{i\neq j}[h_c(i) = h_c(j)] = O(\sigma/n)$, where the probability is taken over the choice of $g_c$.

For each cluster, we define a set of \defn{chains}.
For any $x\in [n]$, the chain of elements starting at $x$ is the sequence of $\sigma$ elements starting with $x$, where each element is obtained by applying $h_c$ to the previous element: 
\[
x, h_c(x), h_c(h_c(x)), \ldots, \underbrace{h_c(h_c(h_c(\ldots h_c}_{\sigma \text{ iterations}}(x)))).
\]
For any $x,i\in [n]$, we say that $x$ \defn{finds} $i$ if $i$ is in the chain of elements starting at $x$.

For each cluster, we sample $n/\sigma^3$ elements $x$ uniformly at random from $[n]$. For each such $x$, we store a key-value pair in a dictionary, where the key is the last element in the chain starting at $x$, and the value is $x$.  We use a linear-space dictionary which achieves constant-time insert in expectation and constant-time worst-case lookup, e.g., a cuckoo hash table suffices, as does~\cite{ArbitmanNaSe10}.  Thus, in $O(1)$ time, for any element $\ell$, we can: (1) determine if $\ell$ is the last element in the chain starting at some sampled element $x$, and (2) if so, find $x$.
A cluster consists of the hash function $g_c$ and this dictionary.

Now, we will describe how to use a cluster to (attempt to) find $f^{-1}(j)$ for some $j\in [n-1]$.  First, we calculate the chain starting at $j$ using $h_c$, exactly as we did when creating the cluster.  For each $x_c$ in the chain starting at $j$, we look up $x_c$ in the dictionary.  If none are found, the query fails.  If we find some $x_c$ in the dictionary, it is the last element in the chain starting at some element $x$; $x$ is stored as the key in the dictionary.  We then calculate the chain starting at $x$. 
If it contains $j$, then the element preceding $j$ in the chain starting at $x$ is $f^{-1}(j)$; we return that element.
If it does not contain $j$, the query fails.  Note that we only perform the above for one $x_c$, so we only calculate two chains.

\begin{lemma}
    \label{lem:function_inversion_cluster_performance}
    Each cluster $c$ can be stored in $O(n/\sigma^3)$ space and created in $O(n (E(f) + \log\sigma)/\sigma^2)$ time.  
    We can query $c$ to attempt to find $f^{-1}(j)$ for some $j\in [n]$ 
    in  $O(\sigma (E(f) + \log\sigma))$ time.
\end{lemma}
\iffull
\begin{proof}
    Storing $g$ requires $O(1)$ amortized space~\cite{FiatNaor00}, and storing the dictionary requires $O(n/\sigma^3)$ space.

    Each time we query, we calculate two chains; this requires $O(\sigma)$ hash table lookups, and requires us to evaluate $h$ $O(\sigma)$ times.
    A table lookup requires $O(1)$ worst-case time, 
    $g_c$ can be evaluated in $O(\log \sigma)$ time, 
    and $f$ can be evaluated in $E(f)$ time.
    Creating a cluster requires creating $O(n/\sigma^3)$ chains.
    \end{proof}
    \fi

The following is the main correctness lemma for function inversion.  
    Its proof is closely adapted from the proof of~\cite[Lemma 4.2]{FiatNaor00}---it is essentially the same proof, but we do not need to condition on having sampled high-indegree elements since we handle them explicitly in $h_c$.  
    \iffull
    We defer this proof to Appendix~\ref{sec:omitted}.
    \fi

 \begin{restatable}{lemma}{functioninversionchainslemma}
\label{lem:function_inversion_chains}
 Let $i$ and $j$ satisfy $f(i) = j$.  Then the probability that $i$ is returned by querying a given cluster $c$ is $\Omega(1/\sigma^2)$.
\end{restatable}

\paragraph{Creating Multiple Clusters.}  To motivate how we set parameters, let us begin with a method that does not quite work.  Consider creating $\Theta(\sigma^2\log \sigma)$ clusters.  For each cluster, we sample a new hash function $g_c$ to construct a new $h_c$, and randomly sample new elements to start the chains.  
The extra $\log \sigma$ term is for Lemma~\ref{lem:function_inversion_all_blocks} below: in short, Fiat-Naor created $\Theta(\sigma^2)$ clusters to ensure each element is returned with constant probability, so 
repeating an extra $\sigma$ times ensures that each element is returned with probability $1 - O(1/\sigma)$.
Unfortunately, this strategy has space usage $\Theta(n(\log \sigma)/\sigma)$ rather than $O(n/\sigma)$.  
Therefore, we adjust parameters to retain the final Lemma~\ref{lem:function_inversion_all_blocks} result while still meeting the $O(n/\sigma)$ space bound; in exchange, the preprocessing and query times will be slightly higher.

We create $\Theta(\sigma^2 \log^3\sigma)$ clusters with parameter $\sigma' := \sigma\log \sigma$.  By Lemma~\ref{lem:function_inversion_cluster_performance}, each cluster requires $O(n/(\sigma\log \sigma)^3)$ space, answers queries in $O(\sigma \log\sigma ( E(f) +\log \sigma))$ time, and can be created in $O(n (E(f)+\log\sigma)/ (\sigma \log \sigma)^2)$ time. By Lemma~\ref{lem:function_inversion_chains}, a given $i\in f^{-1}(j)$ is returned by querying a given cluster with probability $\Omega(1/(\sigma \log \sigma)^2)$.

With this adjustment, we have $O(n/\sigma^3)$ total space over all clusters, and we can give our lemma showing that a given element is likely to be inverted by some cluster.

\begin{lemma}
\label{lem:function_inversion_all_blocks}
   For any $j$, if $i\in f^{-1}(j)$, with probability $1 - O(1/\sigma)$ there exists a cluster $c$ such that there exists a sampled $x$ in cluster $c$ where $x$ finds $i$.
\end{lemma}
\iffull
\begin{proof}
    By Lemma~\ref{lem:function_inversion_chains}, $i$ is returned for a given cluster with probability $\Omega(1/(\sigma\log \sigma)^2)$.  
    Since the clusters are independent, the probability that $i$ is not returned by any cluster is 
    \[
    (1 - \Omega(1/(\sigma\log \sigma)^2))^{\Theta(\sigma^2\log^3 \sigma)}.
    \]
    With a sufficiently large number of clusters 
    (setting the constant in the $\Theta(\sigma^2\log^3 \sigma)$ term to be sufficiently large, 
    and using $(1 - 1/y)^{y} \leq 1/e$ for $y > 1$), 
    this is $O(1/\sigma)$.
\end{proof}
\fi

\subsubsection{Inverting All Points.}
\label{sec:inverting_all_points}
Now, we extend the data structure to obtain \emph{all} points in a given element's preimage.

Lemma~\ref{lem:function_inversion_all_blocks} leaves us with $O(n/\sigma)$ expected $i\in [n]$ that are not inverted by the data structure: 
\iffull
in other words, 
\fi
$O(n/\sigma)$ elements $i$ such that $f(i) = j$ for some $j\in [n]$, but $i$ is not returned when finding $f^{-1}(j)$ using the above strategy.  We call these the \defn{missing items}.

We want to be sure that the missing items will also be returned during a function inversion query.  To this purpose, we store another linear-space dictionary with constant-time lookup and expected linear preprocessing time.  The keys will consist of all $j$ such that there is a missing item in $f^{-1}(j)$.  The value for each such $j$ will be all missing $i$ in $f^{-1}(j)$.  (Note that we only store the missing elements as values, not all elements in $f^{-1}(j)$.)

Now, let's describe how to calculate these values.
The items not in any chain can be found 
while creating the chains, 
by 
keeping an array on $[n]$ and marking all elements found by each chain.  (The preprocessing space can be reduced to $O(n/\sigma)$ in exchange for increasing the preprocessing time by $\sigma$, by repeating this process for each set of $n/\sigma$ elements.)
 Then, we iterate through each $i$ not found in any chain, and calculate $f(i)$ in $O(E(f))$ time.  If the key $f(i)\neq \bot$ is not already a stored in the dictionary, we add it (initialized with the empty set as its value). Then we add $i$ to the set stored as the value for key $f(i)$.

With this extra dictionary for elements that are not inverted by a chain, we add an extra step to the query.
 On a query $j$, we first query the dictionary using $j$ to find any missing items in $f^{-1}(j)$ (this takes time $O(1 + |f^{-1}(j)|)$), and then search for $j$ in the chains.

 We have now proven Theorem~\ref{thm:function_inversion_all}.

\subsection{Putting It All Together}
\label{sec:space_efficient_data_structure_definition}

\paragraph{What to Store.}
Our data structure consists of the following parts.  
First, we store the First $k+1$ Mismatches Data Structure from Lemma~\ref{lem:lcp} built on $T$.
Then, we store the truncated tree $\mathcal{T}$.  
Then, we store the function inversion data structure from Theorem~\ref{thm:function_inversion_all} for each $f_{\lambda}$.

\iffull
In addition, if $\tau = 1$ (i.e., our desired space bound is $\Omega(n)$), 
we store the data structure from Lemma~\ref{lem:fast_query} to improve the search time when $r' = 0$, and we store the data structure from Lemma~\ref{lem:fast_function} to improve the time to evaluate $f_{\lambda}(i)$.
\fi

\paragraph{Queries.}
We describe the query as constructing a list $L$ consisting of leaves of $\calT$; our plan is to find all suffixes contained in $L$ using function inversion, and return all that are close to $q$.  
Note that we do not need to actually store $L$ explicitly as the leaves can be handled one by one during a query
\iffull
---we define the algorithm this way to simplify exposition and analysis.  
\fi

We construct $L$ using the following recursive process:
\begin{itemize}
\item to recurse with $r' > 0$, we use the process outlined in Section~\ref{sec:cgl_queries}, outputting any matching pivots found, and adding any leaf of $\calT$ reached to $L$.
\item  The base case is when $r'=0$ for some altered query $q'$ and node $v\in \calT$. When $r'=0$ we first find the traversal destination and the matching nodes of $v\in \calT$ and $q'$.
If $\tau > 1$, this is done using a manual traversal (Definition~\ref{def:traversal_destination}); if $\tau = 1$ this 
is done using Lemma~\ref{lem:fast_query} (see Section~\ref{sec:improving_logs}).

Once we have found the traversal destination and matching nodes, we process them.
For  each node $v'\in P(v, q', \calT)$, if $v'$ is a leaf of $\calT$, we add $v'$ to $L$; if $v'$ is an internal node of $\calT$, we add its pivot to the output.
\end{itemize}

\iffull
After the above process has completed, we have finished constructing $L$.
\fi
Now, for each leaf in $L$, let $\ell$ be the label of the leaf and $\lambda$ be its path label.  
We use the function inversion data structure from Theorem~\ref{thm:function_inversion_all} to find $f^{-1}_{\lambda}(\ell)$.  
Then, for each $i\in f^{-1}_{\lambda}(\ell)$, we use Lemma~\ref{lem:lcp} to find $d_H(q,s_i)$; if $d_H(q,s_i)\leq r$ then we add $s_i$ to the output.

\iffull
\paragraph{Analysis.}
 To begin, we explain why our query is correct---in particular, why our query on $\calT$ will return all strings returned by the same query on $\calT_0$.
 Note that when an internal node of $\calT$ is traversed, the query is compared to the pivot to see if their distance is at most $r$; when a leaf node $\ell$ is traversed, the query is compared to all elements contained in $\ell$.    
 Let $\calT_f$ be the nodes $v\in \calT_0$ that are
 traversed by $q$ in $\calT_0$.
 It is sufficient to show that the traversal of $\calT$ visits all nodes in $\calT_f\cap \calT$.

\begin{lemma}
\label{lem:truncated_query_correct}
All nodes in $\calT_f\cap \calT$ are traversed during a query. Furthermore, $L$ consists exactly of the leaves of $\calT$ that are in $\calT_f$.
\end{lemma}
\begin{proof}
 This follows immediately from the traversal method.
 Consider traversing $\calT_0$ using the method described in Section~\ref{sec:cgl_queries}.  For recursive calls with $r > 1$, the traversals are identical until we reach a leaf of $\calT$ (which is added to $L$).  For a recursive call on a node $v\in \calT_0$ with $r = 0$, we will by definition find all descendants of $v$ in $\calT$ traversed using a DFS call.  
 \end{proof}

 Let's bound the space and preprocessing time for $\calT$.  Using Lemma~\ref{lem:truncated_num_nodes} and that each node can be stored in  $O(k)$ space, we have that $\calT$ requires  $O(kn\binom{\log n}{k}/\sigma)$ space.  The time to build $\calT$ is upper bounded by the time to build $\calT_0$: 
 $O(n\binom{\log n}{k}k^2\log n)$ time if $\tau = 1$, and 
 $O(n\binom{\log n}{k}k^2\tau\log n)$ time if $\tau > 1$
 (Lemma~\ref{lem:cgl_preprocessing}).
 
The First $k+1$ Mismatches Data Structure requires $O(n/\tau)$ space and $O(n)$ preprocessing time (Lemma~\ref{lem:lcp}).

The function inversion data structure requires $O(n/\sigma)$ space and $O(n (E(f) + \log \sigma) \log \sigma)$ expected preprocessing time for each $\lambda$ (Theorem~\ref{thm:function_inversion_all}).  We multiply by the number of possible $\lambda$, namely, $\binom{\log n}{k}$ (from Observation~\ref{obs:path_label_basics}).  This gives $O(n\binom{\log n}{k} /\sigma)$ total space for all function inversion data structures.
To simplify preprocessing time, we loosely upper bound $E(f)$ by assuming we use a manual traversal: $E(f) = O(k\tau\log n)$; this upper bounds the preprocessing time by $O(n\binom{\log n}{k} k\tau \log n\log \sigma)$. 

If $\tau  = 1$, we also store the data structures from Lemmas~\ref{lem:fast_function} and~\ref{lem:fast_query}.  This requires $O(n\binom{\log n}{k}/\sigma + n)$ space and $O(kn\binom{\log n}{k}/\sigma + kn)$ expected preprocessing time.

 Summing, the expected processing time is
\[
O\left(n\binom{\log n}{k} (k^2\log n + k\tau\log n \log \sigma)\right).
\]

Finally, we bound our query time.  First, we bound the size of $L$.

\begin{lemma}
\label{lem:num_leaves}
    $|L| = O\left( 3^k\binom{\log n}{k} + ({\occ})/\sigma\right)$
\end{lemma}
\begin{proof}
Consider querying $\calT_0$ with $q$.
Recall that $\calT_q$ consists of the nodes of $\calT_0$ that are traversed during the query $q$ with radius $r > 0$, and not during a DFS call.
By Lemma~\ref{lem:truncated_query_correct} and definition of $\calT_q$, each leaf $\ell \in L$  either has $\ell \in \calT_q$; or DFS is called on $\ell$ when querying $\calT_0$ with $q$.  Since $|\calT_q| = O(3^k\binom{\log n}{k})$ by Lemma~\ref{lem:fast_query_calls}, we must only bound nodes we called DFS on.

All $\ell\in L$ are leaves of $\calT$, so each has $\geq \sigma$ descendants in $\calT_0$ by definition of $\calT$; if we called DFS on $\ell$, all of its $\geq \sigma$ descendants are a part of the output (and contribute to $\occ$). 
By Lemma~\ref{lem:no_duplicates}, each descendant contributes exactly once to $\occ$ across all leaves of $\calT$.  Therefore, there can be $\occ/\sigma$ such leaves.
\end{proof}

\begin{lemma}
\label{lem:final_query_time}
If $\tau = 1$, 
    a query $q$ can be answered in time
    \[
    \tilde{O}\left(|q| + 3^k \sigma^3 \binom{\log n}{k} + \sigma^2 \cdot (\occ)\right).
    \]
     If $\tau > 1$, 
    a query $q$ can be answered in time
    \[
    \tilde{O}\left(|q| + 3^k \sigma^3\tau \log n\binom{\log n}{k} + \sigma^2\tau \log n\cdot (\occ)\right).
    \]
\end{lemma}
\begin{proof}
    Recall that $q$ is initialized in $O(|q|)$ time for Lemmas~\ref{lem:lcp} and~\ref{lem:fast_query} if $\tau = 1$, and $O(|q|\log |q|)$ time if $\tau > 1$.

  Let us bound the time to construct $L$.  
  We can traverse a node of $\calT$ in time $O(k)$ if $\tau = 1$, or $O(k(\tau + \log n))$ if $\tau > 1$. 
  While constructing $L$, each node traversed with $r' > 0$ has at least two children; the number of such nodes traversed is bounded by $O(|L|)$.
  For each node with $r' = 0$, if $\tau = 1$ we invoke Lemma~\ref{lem:fast_query} to quickly find the leaf to add to $L$, and if $\tau > 1$ we do a manual traversal of $O(\log n)$ nodes.  
  
   Summing,  if $\tau > 1$, the time to find $L$ is bounded by
   \[
   O\left(|L|\cdot k(\tau + \log n)\log n\right).
   \]
   
   If $\tau = 1$, the time to find $L$ 
   is upper bounded by 
   \[
    O\left(|L|\cdot (k^2 + k\log\log n)\right).
    \]
  
  In each case, we also find all pivots of nodes along the path that have distance at most $k$ from $q$; in $O(k\cdot \occ)$ additional time if $\tau = 1$ and $O(k(\tau + \log n)\cdot\occ)$ if $\tau > 1$.
   
   Each $\ell\in L$ can be inverted in  $O(\sigma^3 k\tau \log n \log^4 \sigma)$ time if $\tau > 1$ or $O(\sigma^3 (k^3 + k^2 \log\log n + \log \sigma)\log^4 \sigma)$ time if $\tau = 1$ using Lemma~\ref{lem:fast_function} and Theorem~\ref{thm:function_inversion_all}.  
   Each time a leaf in $L$ is inverted, each of the $\leq \sigma$ elements must be checked if it is close to $q$ using Lemma~\ref{lem:lcp}, each in time 
   $O(k(\tau + \log n))$ if $\tau > 1$, or
   $O(k)$ if $\tau = 1$.
   
   Summing, the total time required for a query with
   $\tau > 1$ is 
   \[
   O\left(|q|\log |q| + |L|
   \cdot 
   \sigma^3 k\tau\log n \log^4 \sigma + k(\tau + \log n)\cdot \occ\right).
   \]
   
   For $\tau = 1$,
   \[
   O\left(|q|+ |L|\cdot \sigma^3 (k^3 + k^2\log\log n + \log \sigma)\log^4 \sigma  + k\cdot \occ\right).
   \]
   Substituting for $|L|$ using Lemma~\ref{lem:num_leaves}, we obtain the lemma.
\end{proof}
\fi

\section{Improving Final Log Factors}
\label{sec:improving_logs}

\iffull
In this section, 
\fi
we prove the following two lemmas, which improve our query performance by a $\log n$ factor if $\tau = 1$.  

We show how to evaluate $f_{\lambda}(i)$ in $\poly(k,\log\log n)$ time.

\begin{restatable}{lemma}{fastfunctionlem}
\label{lem:fast_function}
    There exists a data structure using $O(n\Smab/\sigma + n)$ space that can be built in $O(kn\binom{\log n}{k}/\sigma + kn)$ expected time such that for any given $\lambda$ and $i$, the data structure can find $f_{\lambda}(i)$ in $E(f_{\lambda}) = O(k^3 + k^2\log\log n)$ time.
\end{restatable}

We show how to answer queries with search radius $r' = 0$ in $\poly(k,\log\log n)$ time.

\begin{restatable}{lemma}{fastquerylem}
\label{lem:fast_query}
There is a data structure using $O(n\Smab/\sigma +n)$ space that can be built using $\calT$ in $O(n\Smab/\sigma +n)$ expected time and can be initialized for $q$ in $O(|q|)$ time such that, for any $k$-altered query $q'$, can find the matching nodes $P(v,q', \calT)$ and the traversal destination $v_T$ of $q'$ and $v$ in 
\iffull
 time 
 \[
 O\left(k^2 + k\log\log n +  k|P(v,q',\calT)|\right).
 \]
 \fi
\end{restatable}

\iffull
\subsection{Summary of Techniques}
\label{sec:fast_prelims}

We begin with a reduction: we show that the key is to be able to quickly find the traversal destination of an altered query or an altered suffix.

Let us explain this reduction for $f_{\lambda}(i)$, that is to say, for Lemma~\ref{lem:fast_function}.  
In short, the challenge is that while $\lambda$ tells us at what depth in the tree the alterations take place, it does not tell us what the alterations are---the alterations themselves depend on the pivot of the respective nodes in the tree.  We use kangaroo jumps (see~\cite{NicolaeRajasekaran17}) to recover the alterations one at a time.

Here is how the kangaroo jumps work.
First, we find the traversal destination of $s_i$, ignoring $\lambda$.  Then, we look at the depth of the first alteration in $\lambda$ (i.e., the index of first character \texttt{a}).  We use a data structure that allows us to to jump to the ancestor node at that depth in constant time.  We then traverse one step in the tree to the \emph{altered} child $c$, and alter $s_i$ to obtain $\hat{s_i}$.  After that, we recurse: we find the traversal destination of $\hat{s_i}$ for node $c$, using the remaining alterations in $\lambda$.  Each time we recurse we handle a new alteration in $\lambda$; after $k+1$ total recursions we are done.

This technique shows that we can build an efficient data structure for Lemma~\ref{lem:fast_function} if we can quickly calculate the traversal destination for an altered string.  In Section~\ref{sec:fast_lemmas_proof} we give this reduction in more detail for both Lemmas~\ref{lem:fast_function} and~\ref{lem:fast_query}.  Then, in Section~\ref{sec:traversal_destination}, we show how to quickly calculate the traversal destination.  Our data structure for calculating the traversal destination, in turn, relies on another subroutine, the ``pivot search data structure,'' which we give in Section~\ref{sec:pivot_search_proof}.

\subsection{Reducing to Finding the Traversal Destination}
\label{sec:fast_lemmas_proof}

In this subsection, we give a proof of Lemmas~\ref{lem:fast_function} and~\ref{lem:fast_query}.  For both lemmas, we assume that we are able to quickly find the traversal destination of any $k$-altered suffix or $k$-altered query; we give the data structure to complete these queries in Section~\ref{sec:traversal_destination} (see Lemma~\ref{lem:traversal_destination_fast} for precise bounds).

\subsubsection{Proof of Lemma~\ref{lem:fast_function}.}
To begin, we point out that (looking ahead) the bounds of Lemma~\ref{lem:traversal_destination_fast} have an extra log factor space, and an extra $\log^2$ factor in preprocessing time---we cannot use $\calT^* = \calT$ as the subtree and still retain our desired bounds.

Therefore, we use the method in Section~\ref{sec:truncating_cgl} (see the discussion before Lemma~\ref{lem:truncated_num_nodes}) to truncate the tree $\calT_0$ to nodes of size $O(\sigma\log^2 n)$.
    We call the resulting tree the \defn{waypoint tree} $\calT^*$.
By Lemma~\ref{lem:truncated_num_nodes}, $\calT^*$ has $O(n/(\sigma/\log^2 n))$ nodes.  
We immediately have that $\calT^*$ is a subtree of $\calT$.
This means we can use $\calT^*$ in Lemma~\ref{lem:traversal_destination_fast} while retaining $O(n\binom{\log n}{k}/\sigma + n)$  space and $O(n\binom{\log n}{k}/\sigma + n)$ preprocessing time.

We store a pointer from each node of the waypoint tree $\calT^*$ to the corresponding node of $\calT$.  We also build a measured-ancestor data structure on the nodes of $\calT^*$; this data structure takes linear space, and given any node in $\calT^*$ and a target depth, allows us to jump to the ancestor of the node at that depth in constant time~\cite{AmirLaLe07}.

Let us briefly describe how our algorithm works; we will give a more detailed description below.
Our goal is to calculate $f_{\lambda}(s_i)$: that is to say, find the node in $\calT$ with path label $\lambda$ that contains $s_i$.
 Our algorithm proceeds in two steps.  First, we go as far as we can in $\calT^{*}$: we find the leaf of $\calT^{*}$ whose path label is a prefix of $\lambda$ and that stores $s_i$.  
We can do this quickly using Lemma~\ref{lem:traversal_destination_fast}.
 Then, we follow the pointer to the corresponding node in $\calT$; we traverse the rest of $\calT$ one node at a time to find $f_{\lambda}(s_i)$.

\paragraph{How to Perform a Query.}
We want to calculate $f_\lambda(i)$ for a given $\lambda$ and $i$ using Lemma~\ref{lem:traversal_destination_fast} built on $\calT^{*}$.

We proceed with the following recursive algorithm.
We maintain a vertex $v$ and a set of alterations $a$; to begin, $v$ is the root of the tree and $a$ is empty.

We use Lemma~\ref{lem:traversal_destination_fast} to find the traversal destination of $(s_i,a)$ for $v$ with respect to $\lambda$.
Let $d$ be the index of the first alteration in $\lambda$; we split into two cases.

If $d$ is less than the depth of $\ell$, 
we jump to the ancestor of $\ell$ with depth $d$; call this ancestor $v_{\alpha}$.  
We pivot-alter $(s_i, a)$ by adding an alteration to $a$. 
Then, we find the altered set that is a child of $v_{\alpha}$ that contains $(s_i, a)$,  
set $v$ equal to this child, remove the first alteration from $\lambda$, and recurse.  If the child does not exist we return $\bot$.

If $d$ is greater than or equal to the depth of $\ell$, we follow the pointer to $\calT$.  
Then, we walk through $\calT$ beginning at $\ell$ to find $f_{\lambda}(i)$.  
By definition, the set $S$ for $\ell$ has size at most $\sigma\log^2 n$; we traverse $\calT$ until reaching a leaf whose set has size at most $\sigma$. By Lemma~\ref{lem:set_size}, this final traversal visits $O(\log\log n)$ nodes of $\calT$.

\paragraph{Correctness.}
Consider calculating $f_{\lambda}(i)$ by traversing $\calT$ using $s_i$.  Let $v_l$ be the last node traversed; by definition, either $v_l$ has label $f_{\lambda}(i)$, or $f_{\lambda}(i) = \bot$.  Notice that during this traversal, for any step without an alteration in $\lambda$, the traversal visits the child whose recursive subset contains the current string---that is to say, the process is identical to the traversals in Definition~\ref{def:traversal_destination}.

We maintain the invariant that after each recursive call, $v$ is on the path from the root to $v_l$.  To begin, $v$ is the root, so the invariant is trivially satisfied.

Consider an iteration beginning at some node $v$ which is on the path from the root to $v_l$.  If $d$ is greater than or equal to the depth of $\ell$, then there are no alterations between $v$ and $\ell$, and $d$ is on the path from the root to $v_l$; the invariant is satisfied.  
If $d$ is less than the depth of $\ell$, we jump to $v_{\alpha}$.  There are no alterations between $v$ and $v_{\alpha}$, so $v_{\alpha}$ is on the path from the root to $v_l$.  We traverse to a child of $v_{\alpha}$; this child must also be on the path from the root to $v_l$ by definition.

After the recursive calls are finished, we traverse $\calT$ to find $f_{\lambda}(i)$.  Since this final traversal starts at an ancestor of $v_l$ by the invariant, we will correctly obtain $f_{\lambda}(i)$.

\paragraph{Analysis.}
Each iteration requires $O(k^2 + k\log\log n)$ time.  
Each time we iterate, we remove an alteration from $\lambda$; since $\lambda$ begins with at most $k$ alterations, 
we obtain $O(k^3 + k^2\log\log n)$ query time.

The space and preprocessing follow immediately from Lemma~\ref{lem:traversal_destination_fast}.

\subsubsection{Proof of Lemma~\ref{lem:fast_query}}
As in the proof of Lemma~\ref{lem:fast_function} above, the bounds of Lemma~\ref{lem:traversal_destination_fast} have an extra $\log$ factor space, and $\log^2$ factor in preprocessing time.
Therefore, we again use the method from Section~\ref{sec:truncating_cgl} to truncate $\calT_0$ to nodes of size $O(\sigma\log^2 n)$; again we call this the \defn{waypoint tree} $\calT^*$.  We build the data structure from Lemma~\ref{lem:traversal_destination_fast} on $\calT^*$.

\paragraph{How to Perform a Query} 

First, we use Lemma~\ref{lem:traversal_destination_fast} to find the traversal destination $v_T$ of $q'$ in $\calT^*$, and to find $P(v,q', \calT^*)$.

We then need to find the rest of $P(v,q', \calT)$ (i.e., the nodes in $\calT\setminus \calT^*)$.  
These are the nodes of $\calT$ traversed by DFS in the tree traversal algorithm.  By definition, the parent of every such node must also be traversed by the tree traversal algorithm.  This means that every node in $P(v,q',\calT)\setminus \calT^*$ either has an ancestor in $P(v,q',\calT^*)$, or is a descendant of $v_T$.  We handle these two cases one at a time.

First, for every leaf of $P(v, q',\calT^*)$, we find the corresponding node in $\calT$, and traverse it using DFS, adding all nodes to $P(v, q',\calT)$.

Then, we follow the pointer from $v_T$ to find the corresponding node in $\calT$.  Beginning at $v_T$, we recursively traverse $\calT$; any nodes traversed by DFS are added to $P(v,q',\calT)$.

\paragraph{Analysis.}  
Initializing Lemma~\ref{lem:traversal_destination_fast} for $q$ requires $O(|q|)$ time; we do this once per query.

Finding the traversal destination and $P(v, q',\calT^*)$ requires $O(k^2 + k\log\log n + k|P(v, q',\calT^*)|)$ time.
By definition of $\calT^*$, node $v$ has $O(\log n)$ descendants in $\calT$; therefore, a manual traversal starting at $v$ visits $O(\log\log n)$ nodes. Each step in the traversal takes $O(k)$ time.  Doing DFS calls from these nodes visits $O(|P(v, q', \calT)|)$ further nodes.  

Adding all terms together, we obtain $O(k^2 + k\log\log n + k |P(v, q',\calT)|)$ total time.
The space and preprocessing follow immediately from Lemma~\ref{lem:traversal_destination_fast}.

\subsection{Finding the Traversal Destination Efficiently}
\label{sec:traversal_destination}

In this section we prove the following lemma.
\begin{lemma}
\label{lem:traversal_destination_fast}
    Let $\calT^*$ be a subtree of $\calT$.  There exists a data structure that requires $O(|\calT^*|\log|\calT^*| + n)$ space and $O(k|\calT^*|\log^2|\calT^*| + kn)$ processing time that can:
    \begin{enumerate}
        \item  Given a $k$-altered suffix $s_i$ of $T$ and a node $v\in \calT^*$, find the traversal destination of $s_i$ rooted at $v$ in $\calT^*$ in $O(k^2 + k\log\log n)$ time.
        \item When a query $q$ arrives, initialize the data structure for $q$ in $O(|q|)$ time.
        \item If already initialized for $q$, then given a $k$-altered query $q'$ and a node $v$, find the traversal destination of $q'$ rooted at $v$ in tree $\calT^*$, and find $P(v, q', \calT^*)$,  in $O(k^2 + k\log\log n + k|P(v, q', \calT^*)|)$ time.
    \end{enumerate}
\end{lemma}

First, in Section~\ref{sec:traversal_destination_setup}, we set up the basics of the data structure---in short, we show that finding the traversal destination is equivalent to finding the predecessor among a specially selected set of altered substrings of $T$.  
In Section~\ref{sec:traversal_destination_proof},
we prove Lemma~\ref{lem:traversal_destination_fast} under the assumption that we can find such predecessors quickly.
Later, in Section~\ref{sec:pivot_search_proof}, we give a data structure that can handle these predecessor queries, completing the proof.

\subsubsection{Setup.}
\label{sec:traversal_destination_setup}
\paragraph{Predecessor Data Structure}
 We repeatedly use a \predds{}~\cite{Willard83} in our construction, which achieves the following bounds.
 For any set $S$ from a known universe of size $n^{O(1)}$, we can preprocess $S$ in $O(|S|\log\log n)$ expected time to create a data structure in $O(|S|)$ space that, for any $i$, can find $pred(i) = \max\{j\in S ~|~ j < i\}$ in $O(\log\log n)$ time.  

\paragraph{Signpost Substrings.}
 A challenge of our data structure---and a structural difference from the approach of~\cite{ColeGoLe04}---is that a string $s$ and its predecessor do not necessarily traverse the tree in the same way.  
 This is true if the predecessor is defined among suffixes of $T$, or even if it is the predecessor among all pivots of $\calT^*$.  
 Essentially, a string and its predecessor need not be in the same recursive subset; they may ``cross the boundary'' into another one.  
 This is because the definition of the recursive subsets is not particularly well-behaved with respect to predecessors.  For example, $S_{< m}$ contains substrings that are particularly high or low in lexicographic order; for another example, two strings in $S_{< \ell}$ have the same longest common prefix with each other as they do to a string in $S_{> \ell}$.  

To fix this issue, we add extra substrings into the predecessor data structure, which we call the \defn{signpost substrings}.  In short, the issue described above is that the predecessor of a string may ``cross the boundary'' into a different subset. If we add the ``boundaries'' themselves to the set, we no longer have this issue.

 Let us explain what the signpost substrings are in more detail; we will construct them explicitly below.
 For any node with set $S$, the recursive sets implicitly partition the strings into $5$ intervals \emph{in lexicographic order}:  observe that all strings in $S_{<\ell}$ are smaller than all strings in $S_{> m}$, which are in turn smaller than all strings in $S_{> \ell}$.  Strings in $S_{< m}$ are either smaller than the strings in $S_{< \ell}$, or larger than the strings in $S_{> \ell}$.  
 The $4$ signpost substrings are exactly the strings on the border between the intervals: so $s_1$ is the smallest string in $S_{< \ell}$ that is not in $S_{< m}$, and so on.  See Table~\ref{table:signpost_substrings}.

 \begin{table*}[t]
 \captionsetup{position=below}
 \begin{center}
 $p = \texttt{BANBNA\$}$ \\
 $m = 3$\\
 
    \renewcommand{\arraystretch}{1.5}
    \bigskip
    \begin{tabular}{p{8ex} p{8ex} p{8ex} p{8ex} p{12ex} p{8ex}  p{10ex}  p{8ex}  p{8ex} }
    \hline
     $S_{< m}$ & $s^1$ & $S_{< \ell}$ & $s^2$ & $S_{> m} $ & $s^3$ & $S_{> \ell}$ & $s^4$ & $S_{< m}$ \\
     \texttt{AN\$} \newline \texttt{BACDA\$}  & \texttt{BAN} & \texttt{BANANA\$} & \texttt{BANB} & \texttt{BANBAAA\$} \newline \texttt{BANBZAAAAA\$}& \texttt{BANC} & \texttt{BANCNAB\$} \newline \texttt{BANCZZ\$}
     & \texttt{BAO} & \texttt{CAT\$} 
     \\
     \hline
    \end{tabular}
    \end{center}
\caption{An example of the signpost substrings.  To emphasize how the signpost substrings help partition the elements, we have also included examples for each recursive subset; $S_{<m}$ has been split into two sets, one with elements $< p$ and one with elements $> p$. Elements in each column are lexicographically larger than those in the previous column.}
 \label{table:signpost_substrings}
\end{table*}

Now, we define each signpost substring explicitly.  The first signpost substring $s^1$ consists of $p$ truncated to the first $m$ characters.  
The second signpost substring $s^2$ consists of $p$ truncated to $m+1$ characters. 
The third, $s^3$, consists of $p$ truncated to $m+1$ characters, with the $m+1$st character incremented.
The fourth, $s^4$, consists of $p$ truncated to $m$ characters, with the $m$th character incremented.    
Each signpost substring is a $(k+1)$-altered substring of $T$; therefore, each can be stored using $O(k)$ space, and they can be used in the $(k+1)$ Mismatches Data Structure in Lemma~\ref{lem:lcp}.
The following lemma immediately follows from the signpost substring definition; the full proof is in the appendix.
\begin{restatable}{lemma}{signpostlemma}
\label{lem:signpost_substrings}
    For any string $s$ and any node $v$,
    let $pred(s, v)$ be the predecessor of $s$ among signpost substrings of descendants of $v$.  Then 
 $pred(s,v)$ and $s$ have the same traversal destination in $\calT^*$ starting at $v$.
\end{restatable}
\begin{proof}
We show that for any descendant $v'$ of $v$, $pred(s,v)$ is in the same recursive set as $s$; this is sufficient to give the lemma.  

For simplicity, in this proof we will use $a \leq b$ for strings $a,b$ to denote that $a$ is no later than $b$ in lexicographic order; otherwise we write $a > b$.  We repeatedly use the following fact: if $a \leq b \leq c$, then the first $LCP(a,c)$ characters of $b$ are the same as those of $a$ and $c$.

Let $s^1$, $s^2$, $s^3$, and $s^4$ be the signpost substrings of $v'$.
We have either $s^1 \leq pred(s,v)$ or $s^1 > s$ by definition of $pred(s,v)$; similarly for $s^2$, $s^3$, and $s^4$.

If $s$ is in $S_{< \ell}$ , we have $s^1 \leq s$, so $s^1 \leq pred(s,v) \leq s$.  We must have that $pred(s,v)$ matches $s^1$ and $s$ for the first $m-1$ characters, and $pred(s,v)[m] \leq s[m] \leq p[m]$; therefore, $pred(s,v) \in S_{< \ell}$.

If $s$ is in $S_{> m}$ , we have $s^2 \leq pred(s,v) \leq s$.  We must have that $pred(s,v)$ matches $s^1$ and $s$ for the first $m$ characters; therefore, $pred(s,v) \in S_{> m}$.

If $s$ is in $S_{> \ell}$, we have $s^3 \leq pred(s,v) \leq s$.  We must have that $pred(s,v)$ matches $s^1$ and $s$ for the first $m-1$ characters, and $p[m] \leq s^3[m] \leq pred(s,v)[m]$; therefore, $pred(s,v) \in S_{> \ell}$.

If $s$ is in $S_{< m}$, we split into two cases.  If $s^4 \leq s$, then $s^4 \leq pred(s,v)$; therefore, $pred(s,v)\in S_{< m}$.
If $s^4 > s$, let $i = LCP(s, p)$; we must have that $s[i] < p[i]$.  Since $pred(s,v) \leq s$, we have either $LCP(pred(s,v),p) < i$ or $pred(s,v)[i] \leq s[i]$; either way, $pred(s,v) \in S_{< m}$.
\end{proof}

\paragraph{Global Index.}
We create a set $P$ consisting of 
all 
pivots of $\calT^*$, and all signpost substrings.
 We sort all elements of $P$, and store them in a sorted array (since they are $(k+1)$-altered substrings of $T$ we can store them implicitly in $O(k)$ space each).  We call each element's position in the sorted array its \defn{global index}.  

 The purpose of the global index is to allow us to do predecessor queries more efficiently.  There are only $O(|P|)$ global indices, so we can do a predecessor query using a \predds{} in $O(\log\log |P|)$ time.

\paragraph{Pivot Search Data Structure.}

Our strategy for Lemma~\ref{lem:traversal_destination_fast} is to repeatedly find the predecessor of our query string among elements of $P$.  
We do this using the Pivot Search Data Structure described 
in Section~\ref{sec:pivot_search_proof}.  
This data structure allows us to find the global index of an altered query or suffix in $O(k^2 + k\log\log n)$ time; see 
Lemma~\ref{lem:pivot_search} for details.

\subsubsection{Proof of Lemma~\ref{lem:traversal_destination_fast}}
\label{sec:traversal_destination_proof}

First, we describe the data structure.
For each node $v$ of $\calT^*$ we store a \predds{} on the global index of all signpost substrings of all descendants of $v$; we call this the \defn{signpost \predds{}}.  We also store a pointer from each index in the \predds{} to its traversal destination.
Then, we store a second \predds{} on the global index of all pivots of descendants of $v$, we call this the \defn{pivot \predds{}}.
We also store all pivots of descendants of $v$ in sorted order so that we can iterate over them in $O(1)$ time.
Recall that we assume that there is a measured-ancestor data structure on $\calT^*$ that allows us to jump from any node to its depth-$d$ ancestor in $O(1)$ time.  We also build the Pivot Search Data Structure from Lemma~\ref{lem:pivot_search} on $|P|$.

Building the signpost \predds{} on all nodes requires a total $O(|\calT^*|\log |\calT^*|\log\log n)$ time since each signpost substring contributes to $O(\log |\calT^*|)$ y-fast tries.  Storing the pointer to each traversal destination of each signpost substring requires $O(|\calT^*|\log |\calT^*|)$ time.  
Building the pivot \predds{} on all nodes similarly requires $O(|\calT^*|\log|\calT^*|\log\log n)$ time.
Building the Pivot Search data structure requires $O(k|P|\log^2|P| + kn) = O(k|\calT^*|\log^2|\calT^*| + kn)$ time.

To initialize the data structure for $q$, we initialize the Pivot Search Data Structure for $q$ (Lemma~\ref{lem:pivot_search}); this requires $O(|q|)$ time. 

To find the traversal destination for a $k$-altered suffix $(s_i, a)$ of $T$ rooted at a vertex $v$, we find the global index $g_i$ of $(s_i, a)$.  Then, we find the predecessor of $g_i$ among signpost substrings of descendants of $v$ using the signpost \predds{} at $v$, and follow the pointer to its traversal destination.  By Lemma~\ref{lem:signpost_substrings}, this must be the traversal destination of $(s_i, a)$.
The same strategy works to find the traversal destination for a $k$-altered query. Finding the global index takes $O(k^2 + k\log\log n)$ time by Lemma~\ref{lem:pivot_search} and finding the predecessor takes $O(\log\log n)$ time; all other operations are $O(1)$.

Finally, note that all pivots in $P(v, q', \calT^*)$ must be a contiguous subset of the pivots of descendants of $v$ in sorted order, and that this subset must be immediately after the predecessor of $q'$.
Thus, to find $P(v, q',\calT^*)$, we find the global index of $q'$ using the Pivot Search Data Structure in $O(k^2 + k\log\log n)$ time.  Then, we find the predecessor of this global index in the pivot \predds{} at $v$.
We iterate through the following descendants of pivots of $v$ in sorted order; for each pivot we find, we use the $k+1$ mismatches data structure to determine if $q$ is a prefix of that pivot.  If so, we add it to $P(v, q', \calT^*)$; if not we are done.  This requires $O(k |P(v, q', \calT^*)|)$ time.

\subsection{The Pivot Search Data Structure}
\label{sec:pivot_search_proof}
 In this section we define the Pivot Search Data Structure, proving the following lemma.
\begin{lemma} 
    \label{lem:pivot_search}
    We can preprocess $P$ in $O(k |P|\log^2 |P| + kn)$ time to create a data structure using $O(|P|\log |P| +n )$ space to
    answer queries of the following form:
    \begin{enumerate}
    \item for any $k$-altered suffix $\hat{s_i}$, we can find the predecessor of $\hat{s_i}$ among elements of $P$ in $O(k^2 + k\log\log n)$ time;
    \item  given a query $q$, we can initialize $q$ in $O(|q|)$ time; then
    \item  after the data structure has been initialized, for any $k$-altered query $q'$, we can find the predecessor of $q'$ among elements of $P$ in $O(k^2 + k\log\log n)$ time.
    \end{enumerate}
\end{lemma}

The Pivot Search Data Structure consists of a compressed trie on $P$, with extra metadata added to aid fast traversal, and with a modification to ensure that its height is $O(\log |P|)$.  This strategy is similar to the ``LCP data structure'' in~\cite{ColeGoLe04}.   We begin by describing the trie modifications; then we describe how to complete each operation listed in Lemma~\ref{lem:pivot_search}.

We point out that all elements of $P$ are $(k+1)$-altered substrings of $T$, so they satisfy the requirements of Lemma~\ref{lem:lcp}.

 \paragraph{Trie Vocabulary.}
We begin with a compressed trie on the entries in $P$.  
We can treat a leaf $\ell$ of the compressed trie, and the corresponding string $p\in P$, interchangeably.   
 
For any trie node $v$, we call the edges along the path from the root to $v$ the \defn{string of $v$}. We say that $v$ \defn{contains} a string $s$ if there is a path from $v$ to a leaf $\ell$ that is a descendant of $v$ where $s$ is equal to the trie labels along that path.  
Note that in this case, $s$ may not be in $P$: rather it is the suffix of some $p\in P$.  
We call $s$ a \defn{suffix rooted at $v$}, or a \defn{suffix} if the node $v$ is not important.

We say that $p\in P$ is a \defn{descendant} of $v$ if the leaf of $p$ is a descendant of $v$.  Therefore, ``descendant'' refers to strings in $P$, whereas ``contains'' refers to suffixes of these strings.  In addition to using this vocabulary, we will make it clear at all points if we are discussing suffixes or elements of $P$.

 The \defn{size} of a tree node $v$ is the number of suffixes it contains; equivalently, the size is the number of $p\in P$ that are descendants of $v$.

\paragraph{Modifying the Trie.}
Consider the following modification of the trie to ensure that it has depth $O(\log |P|)$.
Consider each node $v$ in nondecreasing order of depth. Let $p$ be the median suffix that $v$ contains, and let $c$ be the child of $v$ containing $p$.  We split $c$ into two nodes: one is built recursively on suffixes $\leq p$, and one on suffixes $> p$.
 With this modification, the compressed trie has height $O(\log |P|)$.  We note that using Lemma~\ref{lem:lcp}, we can determine if a $k$-altered query or $k$-altered suffix of $T$ is $\leq p$ or $> p$ in $O(k)$ time, allowing us to traverse the trie as normal.
 
We build a measured-ancestor data structure on the trie~\cite{AmirLaLe07}.  This data structure does not increase the size of the tree and allows us to, given a leaf $\ell$ and a number of characters $d$, find the lowest ancestor of $\ell$ such that the path from the root to that ancestor has $\leq d$ total ancestor labels in $O(1)$ time.  We also store a pointer from each node in the trie to its predecessor.  

\paragraph{The Trie Index.}
Let $S_P$ be the set of all suffixes contained in all nodes of the trie (since the root has all elements of $P$ as a suffix, $P\subseteq S_P$).  We have that $|S_P| = O(|P|\log |P|)$.  We build a suffix array on $T$.  Then,   we build a sorted array on all suffixes in $S_P$, and all suffixes of $T$; the \defn{trie index} of any of these strings is its index in the sorted array.  We store an array so that given $i$, we can find the trie index of any suffix $s_i$ of $T$ in constant time.  
This data structure requires $O(|P|\log |P| + n)$ space.  

\paragraph{Modifying the Compressed Trie.}
For each node $v$ in the compressed trie, we store the largest and smallest $p\in P$ that are descendants of $v$.  
We also store a pointer to the lowest ancestor of $v$ with a descendant smaller than all descendants of $v$.  (We skip this pointer when the smallest $p\in P$ is a descendant of $v$, as then this ancestor does not exist.)
We also store a \predds{} on the trie index of all suffixes contained in $v$.  
We store a pointer from each index in the predecessor data structure to the leaf of that pivot suffix.
Finally (to help traverse when the alphabet is large), we store an additional \predds{} to answer the query: for a given character $c$, which outgoing edge of $v$ begins with the largest character $\leq c$?\footnote{There may be two such edges if the edge was split around the median; the data structure may return either.}

\paragraph{Finding the Predecessor Using the Last Traversal Node.}
Consider traversing the trie starting at a node $v$ using a (possibly altered) string $\frakq$. Let $v_q$ be the last trie node traversed; we call this the \defn{last traversal node}.  
This is equivalent to finding the substring $s'$ contained in $v$ that maximizes $LCP(s', \frakq)$.
We describe how, given the last traversal node, we can find the predecessor of $\frakq$ among suffixes of $v$.

 If $v_q$ is a leaf, then its string is the predecessor of $\frakq$.  
 Otherwise, if $v_q$ is not a leaf, we 
compare $\frakq$ to the smallest descendant of $v_q$.  

If $\frakq$ is larger, the predecessor of $\frakq$ is a descendant of $v_q$.  Since $v_q$ is the last node traversed by $\frakq$, each child of $v_q$ must either (1) have only descendants that are larger than $\frakq$, or (2) have only descendants that are smaller than $\frakq$.  We use the \predds{} stored at $v_q$ to find the outgoing edge of $v$ with the largest character $\leq \frakq[1]$.  The largest descendant of that child is the predecessor of $\frakq$.

If $\frakq$ is smaller, we follow the pointer from $v_q$ to the lowest ancestor with a descendant smaller than all descendants of $v_q$.  Then we proceed as in the case that $\frakq$ is larger to find the predecessor.

\paragraph{Initializing $q$.} 
To initialize $q$, 
we add all suffixes of $q$ to the suffix tree of $T$ in $O(|q|)$ time using~\cite{Ukkonen95}.  With this data structure, for any suffix $q_i$ of $q$, we can find the largest suffix of $T$ that is $\leq q_i$ in $O(1)$ time.

\subsubsection{Finding the Predecessor of an Altered Suffix of $T$}
\label{sec:predecessor_suffix}

In this section we explain how to perform the first type of query from Lemma~\ref{lem:pivot_search}: the predecessor in $P$ of a $k$-altered suffix of $T$.  We use $(s_i, a)$ to refer to this $k$-altered suffix.

The idea is that we will handle alterations in $a$ one at a time: using a predecessor to jump to a node as if the query were not altered, then jumping up to an ancestor to traverse the tree once using the alteration.  
This is the classic ``kangaroo jump'' strategy, which we also use in Lemmas~\ref{lem:lcp} and~\ref{lem:fast_function}.

 \paragraph{Finding the Last Traversal Node with No Alterations.}
Consider the case where $|a| = 0$.  
We find the last traversal node of $s_i$ from some trie node $v$.

Finding the last node in the traversal starting at $v$ is 
equivalent to finding the suffix that maximizes $LCP(p,s_i)$ among pivot suffixes contained in $v$.  The suffix that maximizes $LCP(p, s_i)$ must be the predecessor or the successor of $s_i$ among the suffixes contained in $v$.

Then we can do the following.  Using the trie index of $s_i$ and the \predds{} stored at $v$, we can find the predecessor and successor of $s_i$ among the suffixes contained in $v$.  We find the $LCP$ of each with $s_i$; whichever is longer is $p$.  
We can use the stored pointers to jump to the leaf of $p$; then, using the measured-ancestor data structure, we can jump to the ancestor of $p$ at depth at most $LCP(p, s_i)$ (which we can calculate using Lemma~\ref{lem:lcp}); this is $v_i$.

\paragraph{Recursively Traversing the Tree.}
Now we describe how to find the predecessor when $|a| > 0$.  As stated above, it is sufficient to find the last trie node traversed by the query $(s_i, a)$

Let us describe our recursive algorithm.
To begin, $v$ is the root of the trie.
We begin by ignoring the alterations, and finding the last traversal node $v_i$ for $s_i$ as above.
 If the depth of $v_i$ is smaller than the index of the first alteration in $a$, then the trie traversal from $v$ to $v_i$ is the same for $(s_i, a)$ and $s_i$, and we are done.

Otherwise, we use the measured-ancestor data structure to jump to the ancestor of $v_i$ with depth equal to the index of the first alteration in $a$; let $v'$ be this node.  
We remove the characters in $s_i$ along the path from $v$ to $v'$ (these characters already contributed to the traversal); this can be done by increasing $i$.
We traverse to the child of $v'$ using $(s_i, a)$, again removing any corresponding characters from $s_i$ and $a$.  Traversing to the child requires two steps: we look up the first character in $(s_i, a)$ using the \predds{} of the first characters of the children of $v'$; then, if the child was split because it contains the median, we use Lemma~\ref{lem:lcp} to determine which edge has a longer LCP with $(s_i, a)$.  
If $(s_i, a)$ does not match all characters on the edge from $v'$ to its child, we have reached the last node in the traversal, and we are done.
Otherwise, we set $v$ equal to this child and recurse.

\paragraph{Analysis.}  Each iteration performs a constant number of queries to Lemma~\ref{lem:lcp} (with $\tau = 1$), to the measured ancestor data structure, and to the \predds{}, for $O(k + \log\log n)$ total time.  Over all $k$ iterations, the total time is $O(k^2 + k\log\log n)$.

Now, the preprocessing.  To build a node, we must calculate its median descendant, and build two $y$-fast tries; the time to accomplish this over all nodes is $O(|P|\log |P|\log\log (|P| + n)) = O(|P|\log^2 |P|)$ time.
We can sort $S_P$ in $O(|P|\log^2 |P|)$ time.  We can merge this sorted array with the sorted array on all suffixes of $T$, doing comparisons in $O(k)$ time using Lemma~\ref{lem:lcp}, in $O(k|P|\log^2 |P| + kn)$ time. 

The space required is $O(|P|\log |P| + n)$.

 \subsubsection{Finding the Predecessor of an Altered Query}
\label{sec:predecessor_query}
Now, consider the third operation of Lemma~\ref{lem:pivot_search}: we want to find the predecessor in $P$ of a $k$-altered query, which we denote $q' = (q, a)$.  We will be recursively querying suffixes of $q'$, so we in fact write $q' = (q_i, a)$ for a suffix $q_i$ of $q$.  To begin, $q_i = q$.

 The method for finding the predecessor of an altered query is similar to that of an altered suffix; however, we need an extra step since we do not know (and it is not clear how to efficiently calculate) the trie index of $q_i$.  
 Our plan is as follows. First, we will find the (unaltered) suffix $s_j$ of $T$ that maximizes $LCP(s_j, q_i)$.  We can find this in $O(k)$ time by finding the predecessor and successor of $q_i$ among suffixes of $T$ (as mentioned above, this is possible in $O(1)$ time since we added all suffixes of $q$ to the suffix tree of $T$ during initialization), and returning the one with a larger $LCP$ with $q_i$.  Our plan is then to proceed as we did in Section~\ref{sec:predecessor_suffix}---but now, we must consider the case where $q_i$  and $s_j$ traverse the trie differently.

As before, our plan is a recursive algorithm based on kangaroo jumps.
Let $v$ be a node in the compressed trie, starting with $v$ equal to the root.
 To begin a recursive call, 
let $s_j$ be the (unaltered) suffix of $T$ that maximizes $LCP(s_j, q_j)$.   We begin by finding the last traversal node of $s_j$; call it $v_j$.

The first time when the traversal of $(q_i, a)$ is different from that of $s_j$ is either: (1) when $(q_i, a)$ splits from $q_i$, at the index of the first alteration in $a$, or (2) when $q_i$ splits from $s_j$, at index $LCP(q_i, s_j) + 1$.  We take the minimum of these two terms to find the index of the first split; call this index $d$.  We use the level ancestor data structure to find the ancestor of $v_j$ whose depth is $d$ more than $v$; this node is $v'$.

Then, we proceed as we did in Section~\ref{sec:predecessor_suffix}: we remove characters in $q_i$ that were already used on the path to $v'$; then, we traverse one step to a child $v_c$ of $v'$ using $(q_i, a)$, again removing matched characters; and finally we set $v$ equal to $v_c$ and recurse. 

\paragraph{Analysis.}  

We begin with the following lemma, which will help us bound the number of iterations.  In short, while kangaroo jumps to the first alteration in $a$ can be directly charged to the number of alterations in $a$, we need a way to charge the kangaroo jumps to $LCP(q_i, s_j) + 1$.  The following lemma allows us to charge these jumps to the alterations in $s_j$.

\begin{lemma}
\label{lem:few_pivot_recursions}
   Each time the above process traverses from a node $v'$ to its child $v_c$, for a node $v'$ where $q_i$ splits from $s_j$, all pivots that are descendants of $v_c$ have an alteration on the edge from $v'$ to $v_c$.
\end{lemma}
\begin{proof}
   Let $p = (s_p, a')$ be a suffix of $P$ contained in $v$ that is also a descendant of $v_c$.
   By definition, $LCP(q_i, p) > LCP(q_i, s_j)$, since $p$ and $q_i$ each match all characters on the path from $v$ to $v_c$, whereas $s_j$ matches all characters on the path from $v$ to $v'$, but does not match the characters on the edge from $v'$ to $v_c$.
   If $p$ does not have an alteration on the path from $v$ to $v_c$, 
   $LCP(q_i, s_p) > LCP(q_i, s_j)$, contradicting the definition of $s_j$.
 \end{proof}

Now, we analyze the query time.
Each iteration requires $O(k + \log\log n)$ time. 
    Each time $d$ is the first alteration in $a$, we remove an alteration from $\frakq$; this can only happen $k$ times.  Since each element of $S_P$ is at most $k+1$-altered, by Lemma~\ref{lem:few_pivot_recursions}, $d$ can only be $LCP(q_i, s_j) + 1$ at most $k+1$ times.  In total, the algorithm recurses at most $2k+1$ times, for a total time of $O(k^2 + k\log\log n)$.

The space and preprocessing match the bounds in Section~\ref{sec:predecessor_suffix}.  This proves Lemma~\ref{lem:pivot_search}.

\fi

\section{Acknowledgments}
\label{sec:ack}

This research is supported in part by National Science Foundation grants CAREER-2045641, CNS-2325956, and 
CCF-2103813.
We would like to thank the anonymous reviewers for their helpful comments and feedback on previous versions of this paper.

\bibliographystyle{plain}  
\bibliography{cgl}

\iffull
\appendix
\section{Omitted Proofs}
\label{sec:omitted}

\correctnesscoreobs*

\begin{proof}
For any string $x$ and indices $y$ and $z$, let $x[y..z]$ be the substring of $x$ from $x[y]$ to $x[z]$, and $x[y...]$ be the suffix of $x$ starting at $x[y]$.  Recall that substrings are $1$-indexed.

By definition, $q[1..i] = p[1..i]$ and $s[1..j] = p[1..j]$.
    We split into three cases; all three are essentially the same argument.
\begin{enumerate}
	\item 
    If $i < j$, $i < |q|$, and $j < \min\{|s|, |p|\}$, then $q[i+1] \neq p[i+1]$, but $s[i+1] = p[i+1]$.  
 Note that 
 \[
 d_H(q,s) = d_H(q[1..i],s[1..i]) + d_H(q[i+1],s[i+1]) + d_H(q[i+2...], s[i+2...]).
 \]
 By definition of $i$, $d_H(q[1..i],s[1..i])=0$ and $d_H(q[i+1],s[i+1]) = 1$, so $d_H(q[i+2...], s[i+2...]) = d_H(q,s)-1$.
 By definition, $\hat{q}[i+1] = p[i+1] = s[i+1]$.  The same decomposition as above then gives $d_H(\hat{q}, s) = d_H(q,s)-1$.
	\item if $i > j$, 
    $i < |q|$, and $j < \min\{|s|, |p|\}$, then $q[j+1] = p[j+1]$, but $s[j+1] \neq p[j+1]$.  
 Note that 
 \[
 d_H(q,s) = d_H(q[1..j],s[1..j]) + d_H(q[j+1],s[j+1]) + d_H(q[j+2...], s[j+2...]).
 \]
 Again, $d_H(q[1..j],s[1..j])=0$ and $d_H(q[j+1],s[j+1]) = 1$, so $d_H(q[j+2...], s[j+2...]) = d_H(q,s)-1$, and 
 $\hat{s}[j+1] = p[j+1] = q[j+1]$.   Then $d_H(q, \hat{s}) = d_H(q,s) - 1$; 
	\item if $i = j$, $q[i+1] \neq s[j+1]$,  $i < |q|$, and $j < \min\{|s|, |p|\}$, then $q[i] = p[i] = s[i]$ but $q[i+1] \neq p[i+1]$ and $s[i+1] \neq p[i+1]$.  
 Note that 
 \[
 d_H(q,s) = d_H(q[1..i],s[1..i]) + d_H(q[i+1],s[i+1]) + d_H(q[i+2...], s[i+2...]).
 \]
 By definition of $i$, $d_H(q[1..i],s[1..i])=0$ and $d_H(q[i+1],s[i+1]) = 1$, so $d_H(q[i+2...], s[i+2...]) = d_H(q,s)-1$.
 By definition, $\hat{q}[i+1] = p[i+1] = \hat{s}[i+1]$.  The same decomposition as above then gives $d_H(\hat{q}, \hat{s}) = d_H(q,s)-1$.\qedhere
\end{enumerate}
\end{proof}

\correctnesscorematchobs*

\begin{proof}
First, consider $j \geq |q|$; 
we have $LCP(q, p) = |q|$ and $LCP(s, p) \geq |q|$, so $d_H(s,q) = 0$.

    Now, consider $j < |q|$.  
    We have 
    \[ 
    d_H(s,q) = d_H(s[1..j],q[1..j]) + d_H(s[j+1],q[j+1]) + d_H(s[j+1...],q[j+1...]).
    \]
    Since $i = |q|$, $q[1..j] = p[1..j]$, and since $j < |s|$, $d_H(s[1..j],q[1..j]) = 0$.  We have $q[j+1] = p[j+1]\neq s[j+1]$ since $j < |s|$.  This means that $d_H(s[j+1...],q[j+1...]) = d_H(s,q) - 1$.  Since $\hat{s}[j+1] = p[j+1] = q[j+1]$, we have $d_H(\hat{s},q) = d_H(s,q) - 1$.
\end{proof}

 \pathdistinctlabels*
\begin{proof}
Consider a node $v$ with set $S$ with $(s_i, a)\in S$ for some $a$.  If $(s_i, a)$ is not the pivot of $S$, then $(s_i, a)$ is in exactly one recursive subset $S'$, and $(s_i, a')$ is in exactly one altered subset $\hat{S'}$.  
If $(s_i,a)$ is the pivot, then it is not stored recursively.

The path to a child of $v$ consists of the path to $v$ with either \texttt{u} or \texttt{a} added.  We immediately have that inductively, all nodes at the same depth have distinct path labels; since the length of the path label is exactly the depth, all path labels must be distinct.  Similarly, if $(s_i, a_1)$ is a pivot, no descendants of $v$ contain $(s_i, a_1)$, so no path label is a prefix of another.
\end{proof}
\functioninversionchainslemma*

\begin{proof}
  This proof is closely adapted from the proof of~\cite[Lemma 4.2]{FiatNaor00}---essentially it is the same proof, but we do not need to condition on having sampled high-indegree elements, since we handle them explicitly in $h_c$.

    The key observation behind this proof is the following.  For any $i,j\in [n]$ with $i\neq j$, $\Pr_{g_c}[h_c(i) = h_c(j)] = O(\sigma/n)$.  This follows immediately from the fact that $|f^{-1}(j)| \leq \sigma$ for any $j\in [n-1]$.  The case where $f(i) = \bot$ adds $1/n$ to the probability that two elements collide.

Let's discuss the layout of this proof.
    Let $E_x$ be the event that the $x$th chain in the cluster contains $i$.  
    First, we show that a given chain in the cluster is unlikely to ``loop''; i.e.\ it is unlikely that some element appears twice in the chain.  Then, if the chain does not loop, we show that the contents of chain $x$ consist of $\sigma$ random samples from $[n-1]$; combining these two results allows us to lower bound $\Pr[E_x]$.  Finally, the $E_x$ are not independent; however, using a similar argument we can bound $\Pr[E_x \cap E_{x'}]$ for any $x\neq x'$ so that we can prove the lemma using the inclusion-exclusion principle.

    Let $e_1, e_2, \ldots, e_{\sigma}$ be random variables for the elements in the $x$th chain in the cluster.

    First we show the probability that a chain loops; i.e., contains less than $\sigma$ distinct elements.  
    Each successive $e_j$ is (if there is no loop in the first $j$ elements) uniformly distributed in $[n]$ since $g_c$ is $2\sigma$-independent.  Then the probability that $e_j$ collides with a previous element is, by union bound, at most $O(j\sigma/n)$.  Summing over all elements, we have probability $O(\sigma^3/n)$ that a chain loops.  
    
    Conditioned on the chain not looping,  the probability that a given $e_j$ is $i$ is $1/n$; again these are independent.  
    Using the binomial theorem, the probability that a given $i$ appears is then $1 - (1 - 1/n)^{\sigma} \geq \sigma/n - O(\sigma^2/n^2)$.  Multiplying by the probability the chain does not loop, we obtain $\Pr[E_x] \geq (1 - O(\sigma^3/n))(\sigma/n -  O(\sigma^2/n^2)) \geq \sigma/n - O(\sigma^2/n^2)$. 

    Now, we upper bound $\Pr[E_x \cap E_{x'}]$ for two fixed $x\neq x'$.  Let their elements be $e_1, \ldots, e_{\sigma}$ and $e_1',\ldots, e_{\sigma}'$.  We have $\Pr[E_x \cap E_{x'}] = \Pr[E_x] \cdot \Pr[E_{x'} ~|~ E_x]$.  We upper bound $\Pr[E_{x'} ~|~ E_x]$ by upper bounding the probability that the chains starting at $x$ and $x'$ have \emph{any} element in common.  We do this one element at a time: for any $j$, given no elements in common up to $e_j'$, the probability that $e_j'$ is equal to any of $e_1, \ldots, e_{\sigma}$ is $\leq \sigma^2/n$ since $g_c$ is $2\sigma$-independent.  Therefore, $\Pr[E_x ~|~ E_{x'}] \leq \sigma^3/n$.  Thus, $\Pr[E_x\cap E_{x'}] \leq \Pr[E_x]\cdot \sigma^3/n$.

    Via the inclusion-exclusion principle, we can lower bound the probability some chain in the cluster finds $i$ by $\sum_{x} E_x - \sum_{x,x'} E_x \cap E_{x'}$.  Using that the number of chains is $n/\sigma^3$, and substituting the above bounds, we obtain a lower bound of 
    \[
    \frac{n}{\sigma^3} \Pr[E_x] - \binom{n/\sigma^3}{2}\Pr[E_x]\cdot \sigma^3/n \geq 
    \frac{n}{2\sigma^3} \Pr[E_x] \geq 
    \frac{n}{2\sigma^3} (\sigma/n - O(\sigma^2/n^2)) = \Omega(1/\sigma^2).\qedhere
    \]
\end{proof}

\signpostlemma*

\begin{proof}
Assume by contradiction that there is some node $v'$ with set $S'$ and pivot $p'$ such that $v'$ is a descendant of $v$, and $s$ and $pred(s,v)$ are in different recursive subsets of $S'$.
Let $m'$ be the median LCP of $p$ and any string in $S'$, and let $s^1, s^2, s^3, s^4$ be the signpost substrings of $v'$. 

By definition, we have $s^1 < s^2 < s^3 < s^4$ in lexicographic order.  Therefore strings define a partition on any possible string $s'$ based on how $s'$ is ordered relative to the $s^1, s^2, s^3, s^4$.  Note that by definition, $s$ and $pred(s,v)$ are in the same partition.  For example, if $s^2 \leq s < s^3$, then $s^2 \leq pred(s,v) < s^3$---we cannot have $pred(s,v) > s^3$ since $pred(s,v) \leq s$, and we cannot have $pred(s,v) < s^2$ by definition of $pred(s,v)$.

Therefore, we are left to show that the partition defined by the signpost substrings exactly lines up to the recursive subsets.  Since both are partitions by definition, we only need to show one direction.

First, we show that $s'\in S_{< m}$, we have $s' < s^1$ or $s' \geq s^4$.  Let $i = LCP(s', p')$, so $i < m$.  We split into three cases.
\begin{enumerate}
\item If $i = |s'|$ then since $s^1$ consists of $p$ truncated to $m$ characters, $s'$ is a prefix of $s^1$, and therefore $s' < s^1$.
\item If $s'[i+1] > p'[i+1]$, recall that by definition $LCP(p', s^4) = m-1$.  If $i+1 \leq m - 1$ we have that $s^4$ and $p'$ match for $i+1$ characters so $s' > s^4$, if $i+1 = m-1$ we have that $s^4$ and $p'$ match for $i$ characters, and $s^4[i+1]$ is the smallest character larger than $p'[i+1]$, so $s' \geq s^4$.
\item If $s'[i+1] < p'[i+1]$, recall that $s^1$ is $p'$ truncated to $m$ characters.  Since $i+1 \leq m$, $LCP(s', s^1) = i$ and $s'[i+1] < s^1[i+1]$, so $s' < s^1$.
\end{enumerate}

Now, we show that if $s' \in S_{< \ell}$, we have $s^1\leq s' < s^2$.  By definition, $LCP(s', p') = m$, and if $s'$ has length $> m$, then $s'[m+1] < p'[m+1]$.  Since $s^1$ consists of $p'$ truncated to the first $m$ characters, $s^1$ is a prefix of $s'$, so $s^1 \leq s'$.  Since $s^2$ consists of $p'$ truncated to the first $m+1$ characters, 
$s^2$ matches $s'$ for $m$ characters, and either $s'$ has length $m$, or $s'[m+1] < s^2[m+1]$; either way, 
$s' < s^2$.

Now, we show that if $s'\in S_{> m}$, we have $s^2 \leq s' < s^3$.  Since $s^2$ consists of $p'$ truncated to the first $m$ characters, $s^2$ is a prefix of $s'$, so $s^2 \leq s'$.  Since $s^3$ consists of $p'$ truncated to the first $m+1$ characters, with the $m+1$st character incremented, we must have that $s'$ and $s^3$ match for the first $m$ characters, and $s'[m+1] = p[m+1] < s^3[m+1]$, so $s' < s^3$.

Finally, we show that if $s' \in S_{> \ell}$, we have $s^3\leq s' < s^4$.  
By definition, $LCP(s', p') = m$, and if $s'$ has length $> m$, then $s'[m+1] > p'[m+1]$.  Since $s^3$ consists of $p'$ truncated to the first $m+1$ characters with the $m+1$st character incremented, $LCP(s', s^3) = m$, and $s^3[m+1] \leq s'[m+1]$, so $s^3 \leq s'$.  Since $s^4$ consists of $p$ truncated to the first $m$ characters, with the $m$th character incremented, $LCP(s', s^4) = m-1$, and $s'[m] < s^4[m]$, so $s' < s^4$.
\end{proof}

\section{Discussion of Longest Common Prefix Data Structures}
\label{sec:lcp_discussion}

For the sublinear case (i.e., $\tau > 1$), when comparing an altered suffix of $T$ to an altered suffix of $q$, we used the Monte-Carlo data structure of Bille et al.~\cite{BilleGoKn15} to answer longest common prefix (LCP) queries.  In this section, we discuss why we chose this data structure rather than more recent and more performant data structures, and explain in more detail why the LCP data structures for $T$ and $q$ can be merged in $O(|q|)$ time.

\paragraph{Merging LCP Data Structures.}
Consider a new string $M$ obtained by concatenating $T$ and $q$.  Recall that $T$ ends with at least one occurrence of a special character $\$$ that does not appear in $T$ or $q$; therefore, the LCP of a suffix of $T$ and $q$ is the same as the LCP of the corresponding indices in $M$.  

With this in mind, our goal is to begin with Bille et al.'s data structure on $T$; then, given $q$, build the data structure on $M$ in $O(|q|)$ time.  

We briefly describe how Bille et al.'s data structure works.  
(This is the ``Monte Carlo data structure'' from~\cite[Section 3.3]{BilleGoKn15}.) They split the string into blocks of length $\tau$ ($\tau$ is the space-saving parameter to obtain $O(n/\tau)$ space---we will use it using the same value as $\tau$ elsewhere in the paper).  Then, for block $x$, they sample $b_x$ evenly spaced prefixes, called the ``sampled positions,'' where $b_x$ is a function of $x$ and $\tau$, and store the Karp-Rabin fingerprint of each.  (In particular, for each sampled position $i$, they store the Karp-Rabin fingerprint of the first $i$ characters of $T$.)

Since $b_x$ is a function of $x$ and $\tau$, we can immediately obtain the sampled positions of $q$ after it is appended to $T$.  If we store the Karp-Rabin fingerprint of $T$, we can find the fingerprint of all prefixes of $M$ in $O(|q|)$ additional time; storing the fingerprints at all sampled positions according to each $b_x$ obtains the combined data structure.

\paragraph{What We Want Out of an LCP Data Structure.}
Before discussing why we chose the data structure we did, we should briefly discuss what performance bounds we need out of the data structure, which can be improved, and how such improvements would affect our final bounds in Theorem~\ref{thm:result_small_space}.

Our LCP data structure needs to have $O(n/\tau)$ space to fit the space bounds of Theorem~\ref{thm:result_small_space}.  Within this space, we want the following; currently we achieve all but the last.
\begin{itemize}
    \item \textbf{Preprocessing time:} We require $O(n\binom{\log n}{k})$ time to build the function inversion data structures using Theorem~\ref{thm:function_inversion_all} (in total for all $\lambda$; see Section~\ref{sec:space_efficient_data_structure_definition}), so the LCP data structure should have preprocessing time of $O(n\binom{\log n}{k})$.  This rules out many previous results with $O(n^c)$ preprocessing time for $c > 1$, but e.g., $O(n\log n)$ and $O(n)$ time preprocessing work equally well for Theorem~\ref{thm:result_small_space}.
    \item \textbf{Efficient Merging:}  After building the LCP data structure on $T$, after $q$ arrives, we need to be able to answer LCP queries between any suffix of $T$ and any suffix of $q$ after $\tilde{O}(|q|)$ further preprocessing.
    \item \textbf{Query time:} In~\cite{BilleGoKn15} the query time is $O(\tau + \log n)$
    rather than the optimal $O(\tau)$.  This winds up being a lower-order term in our analysis.  In short, this is because we use~\cite{BilleGoKn15} only to either traverse the tree using a query $q$, or to find the distance from $q$ to a suffix of $T$.  Each time we find the distance from $q$ to a suffix of $T$, we spent $\Omega(\sigma^2)$ work obtaining the suffix of $T$ using the function inversion data structure; since $\tau > 1$ implies $\sigma > \log n$, this is a lower-order term.  Traversing the tree similarly winds up being a lower-order term.
    \item \textbf{Query correctness:} This is the area for improvement.
    In~\cite{BilleGoKn15}, the query is correct with high probability.  If we could instead use an LCP data structure that is correct with probability $1$, then Theorem~\ref{thm:result_small_space} would---like Theorem~\ref{thm:result_large_space}---give the correct answer with probability $1$.
\end{itemize}

\paragraph{Other Options for Longest Common Prefix Data Structures.}
Bille et al.~\cite{BilleGoKn15} give Las Vegas and deterministic variants of their data structure. However, they are at the cost of larger preprocessing or query time.
Similarly, Gawrychowski and Kociumaka~\cite{GawrychowskiKociumaka17} give a data structure with $O(\tau)$ query time and queries correct with probability $1$, but $O(n^{3/2})$ preprocessing time.

More recently, Birenzwige, Golan, and Porat~\cite{BirenzwigeGoPo20} gave a data structure with $O(\tau)$ queries that are always correct, and $O(n)$ preprocessing time with high probability.  Finally, Kosolobov and Sivukin~\cite{KosolobovSivukin24} give an optimal, fully deterministic data structure with $O(n)$ preprocessing and $O(\tau)$ query.

The challenge for using these recent data structures is that we want merges to be efficient.  This does not seem to be the case, at least out of the box.

To briefly describe the challenge, these results generally contain, as a subroutine, a sorting step.  This is not usually a problem if merges are not a goal: all $n$ suffixes of $T$ can be sorted in $O(n)$ time using standard techniques, and all suffixes of $q$ can be sorted in $O(|q|)$ time similarly.  However, merging the two sorted lists entails placing the suffixes of $q$ in sorted order with respect to the suffixes of $T$, which requires $\Theta(|q|\log n)$ time without further structure (which would affect the bound in Theorem~\ref{thm:result_small_space}).  See for example the proof of~\cite[Lemma 3]{KosolobovSivukin24}, or the process of converting blocks of $S$ into unique symbols (and indeed all uses of Lemma 2.4) in~\cite{BirenzwigeGoPo20}.

It is entirely possible that these challenges can be overcome---for example, Ukkonnen's algorithm~\cite{Ukkonen95} immediately allows us to merge the suffix array of $T$ and $q$ in $O(|q|)$ time; this strategy was used e.g., in~\cite{ColeGoLe04}.
 Alternatively, it seems possible that we could use~\cite{KosolobovSivukin24} as our LCP data structure with a $O(|q|\log n)$ time merge algorithm as sketched above;  if this works, we would obtain an equivalent of Theorem~\ref{thm:result_small_space} with $\tilde{O}\left(|q|\log n + 3^k\sigma^3 \tau  \log n \binom{\log n}{k} +  \sigma^2 \tau \log n(\occ)\right)$ time queries that are always correct.  
  We leave these avenues of improvement for future work.

\fi

\end{document}